\documentclass{llncs}

\usepackage[ruled,lined]{algorithm2e}% provides Algorithm environment

\usepackage{amssymb}
\usepackage{graphicx}
\usepackage{color}
\usepackage{amsmath}
\usepackage{amssymb}
\usepackage{xspace}
\usepackage{enumerate}
\usepackage{comment}

\usepackage{url}
\newcommand{\slength}[1]{|#1|}
\newcommand{\length}{\ensuremath{\mathit{len}}}
\newcounter{instr}

\newcommand{\ep}{end\text{-}pos}
\newcommand{\suf}{s\ell}  % suffix link function
\newcommand{\isuf}{s\ell^*}  % improved suffix link function
\newcommand{\aut}{\mathcal{A}}

\newcommand{\R}{\mathcal{R}\hspace{-1.2pt}_{_{x}}}
\newcommand{\notR}{\not\hspace{-3pt}\mathcal{R}\hspace{-2pt}_{_{x}}}
\newcommand{\bigO}{\mathcal{O}}
\newcommand{\Suff}{\mathit{Suff}\xspace}
\newcommand{\Fact}{\mathit{Fact}\xspace}
\newcommand{\myroot}{\mathit{root}\xspace}

\newcommand{\nil}{\textsf{null}\xspace}

\newcommand{\defAs}
{=}

\newcommand{\trdsuff}
{\stackrel{\textit{\tiny utd}}{\sqsupseteq}}

\newcommand{\ttrdsuff}
{\stackrel{\textit{\tiny $t$-utd}}{\sqsupseteq}}

\newcommand{\trdmatch}
{\stackrel{\textit{\tiny utd}}{=}}

\newcommand{\val}{\mathit{val}\xspace}
\newcommand{\Prob}{\mathit{Pr}\xspace}

\newlength{\gnat}

\newlength{\gnatb}

\newcommand{\X}{X}
\newcommand{\Y}{Y}
\newcommand{\Z}{Z}

\newcommand{\expect}{E}

\newcommand{\md}{\mathit{\mu}\xspace}

\begin{document}

\title{Text Searching Allowing for Non-Overlapping Adjacent Unbalanced Translocations\thanks{Preliminary versions of the results presented in this article have been presented in two conference papers \cite{FP19,CFP20}}}

\author{Domenico Cantone$^{\dagger}$, Simone Faro$^{\dagger}$ and Arianna Pavone$^{\ddagger}$}

\institute{
$^{\dagger}$Department of Mathematics and Computer Science \\
University of Catania, Viale A. Doria n.6, 95125, Catania, Italy
\email{\{domenico.cantone,simone.faro\}@unict.it}\\[0.2cm]
$^{\ddagger}$Department of Cognitive Science\\University of Messina, via Concezione n.6/8, 98122, Messina, Italia\\
\email{apavone@unime.it}
}

\maketitle

\begin{abstract}
In this paper we investigate the \emph{approximate string matching problem} when the allowed edit operations are \emph{non-overlapping unbalanced translocations of adjacent factors}. 
Such kind of edit operations take place when two adjacent sub-strings of the text swap, resulting in a modified string. The two involved substrings are allowed to be of different lengths. 

Such large-scale modifications on strings have various applications. They are among the most frequent chromosomal alterations, accounted for 30\% of all losses of heterozygosity, a major genetic event causing inactivation of cancer suppressor genes. In addition, among other applications, they are frequent modifications accounted in musical or in natural language information retrieval. 
However, despite of their central role in so many fields of text processing, little attention has been devoted to the problem of matching strings allowing for this kind of edit operation.

In this paper we present three algorithms for solving the problem, all of them with a $\bigO(nm^3)$ worst-case and a $\bigO(m^2)$-space complexity, where $m$ and $n$ are the length of the pattern and of the text, respectively.
In particular, our first algorithm is based on the dynamic-programming approach. Our second solution improves the previous one by making use of the Directed Acyclic Word Graph of the pattern. 
Finally our third algorithm is based on an alignment procedure. 
We also show that under the assumptions of equiprobability and independence of characters, our second algorithm has a $\bigO(n\log^2_{\sigma} m)$ average time complexity, for an alphabet of size $\sigma \geq 4$.
\end{abstract}


%%%%%%%%%%%%%%%%%%%%%%%%%%%%%%%%%%%

%%%%%%%%%%%%%%%%%%%%%%%%%%%

\section{Introduction}\label{sec:Introduction}

\emph{String alignment} and \emph{approximate string matching} are two fundamental problems in text processing. Given two input sequences $x$, of length $m$, and $y$, of length $n$, the \emph{string alignment} problem consists  in finding a set of edit operations able to transform $x$ in $y$, while the \emph{approximate string matching} problem consists in finding all approximate matches of $x$ in $y$.  The closeness of a match is measured in terms of the sum of the costs of the elementary edit operations necessary to convert the string into an exact match.

Most string matching methods are based on the \emph{Levenshtein distance}~\cite{Lev66}, commonly referred to just as \emph{edit distance}, or on the \emph{Damerau distance}~\cite{Dam64}. The edit operations in the case of the Levenshtein distance are \emph{insertions}, \emph{deletions}, and \emph{substitutions} of characters, whereas, in the case of the Damerau distance, \emph{swaps} of characters, i.e., transpositions of two adjacent characters, are also allowed (for an in-depth survey on approximate string matching, see~\cite{Nav01}).  Both distances assume that changes between strings occur locally, i.e., only a small portion of the string is involved in the mutation event.  However evidence shows that in many applications there are several circumstances where large scale changes are possible~\cite{CHK15,CH03,VAL06}. For instance, such mutations are crucial in \textsc{dna} since they often cause genetic diseases \cite{Lupski98,Oliver02}. 
For example, large pieces of \textsc{dna} can be moved from one location to another (\emph{translocations}) \cite{CHK15,OK08,War91,WHR15}, or replaced by their reversed complements (\emph{inversions}) \cite{CCF13}. 

Translocations can be \emph{balanced} (when equal length pieces are swapped) or \emph{unbalanced} (when pieces with different lengths are moved).
Interestingly, unbalanced translocations are a relatively common type of mutation and a major contributor to neurodevelopmental disorders \cite{WHR15}.
In addition, cytogenetic studies have also indicated that unbalanced translocations can be found in human genome with a de novo frequency of 1 in 2000 \cite{War91} and that it is a frequent chromosome alteration in a variety of human cancers \cite{OK08}. Hence the need for practical and efficient methods for detecting and locating such kind of large scale mutations in biological sequences arises.

\subsection{Related Results}
In the last three decades much work has been made for the alignment and matching problem allowing for chromosomal alteration, especially for non-overlapping inversions. Table \ref{tab:related} shows the list of all solutions proposed over the years, together with their worst-case, average-case and space complexities.

\begin{table}[t!]
\begin{small}
\begin{center}
\begin{tabular}{llllll}
&   \textsc{Authors}~~~~ & \textsc{Year}~~~~ & \textsc{W.C. Time}~~~~ & \textsc{AVG Time}~~~~  &  \textsc{Space}\\
\hline
\multicolumn{6}{l}{\textbf{Alignment with inversions}}\\%[0.2cm]
&   Schoniger and Waterman \cite{SW92} & (1992) & $\bigO(n^2m^2)$ & - & $\bigO(m^2)$\\
&   Gao \emph{et al.} \cite{Gao03} & (2003) & $\bigO(n^2m^2)$ & - & $\bigO(nm)$\\
&   Chen \emph{et al.} \cite{CGLNWW04}  & (2004) & $\bigO(n^2m^2)$ & - & $\bigO(nm)$\\
&   Alves \emph{et al.} \cite{ADV05}  & (2005) & $\bigO(n^3\log n)$ & - & $\bigO(n^2)$\\
&   Vellozo \emph{et al.} \cite{VAL06} & (2006) & $\bigO(nm^2)$ & - & $\bigO(nm)$\\%[0.2cm]
\hline
\multicolumn{6}{l}{\textbf{Alignment with inversions and balanced translocations on both strings}}\\%[0.2cm]
&   Cho \emph{et al.} \cite{CHK15} & (2015) & $\bigO(m^3)$ & - & $\bigO(m^2)$\\%[0.2cm]
\hline
\multicolumn{6}{l}{\textbf{Pattern matching with inversions}}\\%[0.2cm]				
&   Cantone \emph{et al.} \cite{CCF11} & (2011) & $\bigO(nm)$ & - & $\bigO(m^2)$\\
&   Cantone \emph{et al.} \cite{CCF13} & (2013) & $\bigO(nm)$ & $\bigO(n)$ & $\bigO(m^2)$\\%[0.2cm]
\hline
\multicolumn{6}{l}{\textbf{Pattern matching with unbalanced translocations}}\\%[0.2cm]
&   Faro and Pavone \cite{FP19} & (2019) & $\bigO(m^2)$ & $\bigO(n)$ & $\bigO(m)$\\%[0.2cm]
\hline
\multicolumn{6}{l}{\textbf{Pattern matching with inversions and balanced translocations}}\\%[0.2cm]
&   Cantone \emph{et al.} \cite{CFG10} & (2010) & $\bigO(nm^2)$ & $\bigO(n\log m)$ & $\bigO(m^2)$\\
&   Grabowski \emph{et al.} \cite{GFG11} & (2011) & $\bigO(nm^2)$ & $\bigO(n)$ & $\bigO(m)$\\
&   Cantone \emph{et al.} \cite{CFG14} & (2014) & $\bigO(nm^2)$ & $\bigO(n)$ & $\bigO(m)$\\%[0.2cm]
\hline
\multicolumn{6}{l}{\textbf{Pattern matching with unbalanced translocations}}\\%[0.2cm]
&   This paper (Algorithm \ref{fig:code1})& (2020) & $\bigO(nm^3)$ & - &   $\bigO(m^2)$\\%[0.2cm]
&   This paper (Algorithm \ref{fig:code2})& (2020) & $\bigO(nm^3)$ & $\bigO(n\log^2 m)$ &   $\bigO(m^2)$\\%[0.2cm]
&   This paper (Algorithm \ref{fig:code3})& (2020) & $\bigO(nm^3)$ & - & $\bigO(m^2)$\\%[0.2cm]
\hline
\end{tabular}
\end{center}
\end{small}
\caption{\label{tab:related}Results related to alignment and matching of strings allowing for inversions and translocations of factors. Unless otherwise specified, all the edit operations allowed by all the listed solutions are intended to involve only non-overlapping factors of the pattern.}
\end{table}

Concerning the alignment problem with inversions, a first solution based on dynamic programming, was proposed by Sch\"oniger and Waterman \cite{SW92}, which runs in $\bigO(n^2 m^2)$-time and $\bigO(n^2m^2)$-space on input sequences of length $n$ and $m$. Several other papers have been devoted to the alignment problem with inversions. The best solution is due to Vellozo \emph{et al.} \cite{VAL06}, who proposed a $\bigO(nm^2)$-time and $\bigO(nm)$-space algorithm, within the more general framework of an edit graph.

Regarding the alignment problem with translocations, Cho \emph{et al.} \cite{CHK15} presented a first solution for the case of inversions and translocations of equal length factors (i.e., balanced translocations), working in $\bigO(n^3)$-time and $\bigO(m^2)$-space. However their solution generalizes the problem to the case where edit operations can occur on both strings and assume that the input sequences have the same length, namely $|x|=|y|=n$.

Regarding the approximate string matching problem, a first solution was presented by Cantone \emph{et al.} \cite{CCF11}, where the authors presented an algorithm running in $\bigO(nm)$ worst-case time and $\bigO(m^{2})$-space for the approximate string matching problem allowing for non-overlapping inversions. Additionally, they also provided a variant \cite{CCF13} of the algorithm which has the same complexity in the worst case, but achieves $\bigO(n)$-time complexity on average.
Cantone \emph{et al.} also proposed in~\cite{CFG10} an efficient solution running in $\bigO(nm^2)$-time and $\bigO(m^2)$-space for a slightly more general problem, allowing for balanced translocations of adjacent factors besides non-overlapping inversions. The authors improved their previous result in~\cite{CFG14} obtaining an algorithm having $\bigO(n)$-time complexity on average.
We mention also the result by Grabowski \emph{et al.}~\cite{GFG11}, which solves the same string matching problem in $\bigO(nm^2)$-time and $\bigO(m)$-space, reaching in practical cases $\bigO(n)$-time complexity.

%%%%%%%%%%%%%%%%%%%

\subsection{Our Results}
While in the previous results mentioned above it is intended that a translocation may take place only between balanced factors of the pattern,
in this paper we investigate the approximate string matching problem under a string distance whose edit operations are non-overlapping unbalanced translocations of adjacent factors. To the best of our knowledge, this slightly more general problem has never been addressed in the context of approximate pattern matching.

Given a bound $\delta$, we consider the following variants of the approximate string matching problem allowing for non-overlapping unbalanced translocations of adjacent factors: 
\begin{enumerate}[(a)]
\item find the number of all unbounded approximate occurrences of $x$ in $y$; 
\item find the number of all $\delta$-bounded approximate occurrences of $x$ in $y$; 
\item find all positions $s$ in $y$ such that $x$ has a $\delta$-bounded approximate occurrence in $y$ at position $s$; 
\item for each position $s$ in $y$, find the number of distinct $\delta$-bounded approximate occurrences of $x$ in $y$ at position $s$.
\end{enumerate}

Specifically, after introducing some useful notions and basic definitions (Section \ref{sec:notions}), we will present the following three solutions for solving variants (a), (b) and (c):

\begin{itemize}
\item First we propose (Section \ref{sec:DP}) a solution to the problem, based on the general dynamic programming approach, which needs $\bigO(nm^3)$-time and $\bigO(m^2)$-space.
%
% automaton based solution
\item Subsequently we propose (Section \ref{sec:automaton}) a second solution to the problem  that makes use of the Directed Acyclic Word Graph of the pattern and achieves a $\bigO(n\log_{\sigma}^2 m)$-time complexity on average, for an alphabet of size $\sigma\geq 4$, still maintaining the same complexity, in the worst case, as for our first solution.
%
% heuristics based solution
\item Finally we present (Section \ref{sec:heuristics}) an alternative searching algorithm for the same problem also working in $\bigO(nm^3)$ worst case time and $\bigO(m^2)$-space but based on an alignment approach. 
\end{itemize}

\section{Basic notions and definitions}\label{sec:notions}
A string $x$ of length $m \geq 0$, over an alphabet $\Sigma$, is represented as a finite array $x[1\,..\,m]$ of elements of $\Sigma$. We write $\slength{x} = m$ to indicate its length.  In particular, when $m=0$ we have the empty string $\varepsilon$.  We denote by $x[i]$ the $i$-th character of $x$, for $1\leq i\leq m$.  Likewise, the
factor of $x$, contained between the $i$-th and the $j$-th characters of $x$ is indicated with $x[i\,..\,j]$, for $1\leq i \leq j \leq m$.  The set of factors of $x$ is
denoted by $\Fact(x)$ and its size is $\bigO(m^2)$.  

A string $w\in \Sigma^*$ is a suffix of $x$ (in symbols, $w \sqsupseteq x$) if $w = x[i\,..\,m]$, for some $1\leq i \leq m$. We denote by $\Suff(x)$ the set of the suffixes of $x$.  Similarly, we say that $w$ is a prefix of $x$ if $w = x[1\,..\,i]$, for some $1\leq i \leq m$.  Additionally, we use the symbol $x_i$ to denote the prefix of $x$ of length $i$ (i.e.,$x_i = x[1\,..\,i]$), for $1\leq i\leq m$, and make the convention that $x_{0}$ denotes the empty string $\varepsilon$.  In addition, we write $xw$ for the concatenation of the strings $x$ and $w$.

For $w\in \Fact(p)$, we denote with $\ep(w)$ the set of all positions in $x$ at which an occurrence of $w$ ends; formally, we put $\ep(w) := \{i\ |\ w \sqsupseteq x_{i}\}$.  For any given pattern $x$, we define an equivalence relation $\R$ 
by putting, for all $w,z\in \Sigma^*$, 
$$w \mathrel{\R} z \iff \ep(w)=\ep(z).$$ 
We also denote with $\R(w)$ the equivalence class 
of the string $w$. For each equivalence class $q$ of $\R$, we put $\length(q) = |\val(q)|$, where $\val(q)$ is the longest string $w$ in the
equivalence class $q$. 

\begin{example}\label{ex1}
Let $x=agcagccag$ be a string over $\Sigma=\{a,g,c,t\}$ of length $m=9$.  Then we have $\ep(ag)=\{2,5,9\}$, since the substring $ag$ occurs three times in $x$, ending at positions $2$, $5$ and $9$, respectively, in that order. Similarly we have $\ep(gcc)=\{7\}$.
Observe that $\R(ag)=\{ag,g\}$.  Similarly we have
$\R(gc)=\{agc,gc\}$.  Thus, we have $\val(\R(ag))=ag$, $\length(\R(ag))=2$, $\val(\R(gc))=agc$ and $\length(\R(gc))=3$.
\end{example}

The Directed Acyclic Word Graph~\cite{CR94} of a pattern $x$ (\textsc{Dawg}, for short) is the deterministic automaton $\mathcal{A}(x)=(Q,\Sigma,\delta,\myroot,F)$ whose language is $\Fact(x)$, where $Q=\{\R(w) : w\in \Fact(x)\}$ is the set of states, $\Sigma$ is the alphabet of the characters in $x$, $\myroot=\R(\varepsilon)$ is the initial state, $F=Q$ is the set of final states, and $\delta:Q\times \Sigma\rightarrow Q$ is the transition function defined by $\delta(\R(y),c) \defAs \R(yc)$, for all $c\in \Sigma$ and
$yc\in \Fact(x)$.

We define a failure function, $\suf:\Fact(x) \setminus \{\varepsilon\} \rightarrow \Fact(x)$, called \emph{suffix link}, by putting, for any $w \in \Fact(x) \setminus \{\varepsilon\}$,
$$\suf(w) = \text{``\,longest } y \in \Suff(w) \text{ such that } y \notR w\text{"}.$$
The function $\suf$ enjoys the following property $$w \mathrel{\R} y \Longrightarrow \suf(w)=\suf(y).$$  We extend the functions $\suf$ and $\ep$ to $Q$ by putting $\suf(q) := \R(\suf(\val(q)))$ and $\ep(q) = \ep(\val(q))$, for each $q \in Q$. Figure \ref{fig:dawg} shows the \textsc{Dawg} of the pattern $x=aggga$, where the edges of the automaton are depicted in black while the suffix links are depicted in red.

\smallskip

A \emph{distance} $d:\Sigma^{*}\times \Sigma^{*}\rightarrow \mathbb{R}$ is a function which associates to any pair of strings $x$ and $y$ the minimal cost of any finite sequence of edit operations which transforms $x$ into $y$, if such a sequence exists, $\infty$ otherwise.  

In this paper we consider the \emph{unbalanced translocation distance}, $utd(x,y)$, whose unique edit operation is the \emph{translocation} of two adjacent factors of the string, with possibly different lengths.
Specifically, in an \emph{unbalanced translocation} a factor of the form $zw$ is transformed into $wz$, provided that both $\slength{z},\slength{w} > 0$ (it is not necessary that $\slength{z}=\slength{w}$). We assign a unit cost to each translocation.

\begin{example}\label{ex2}
Let $x=g\overline{t}\underline{ga}c\overline{cgt}\underline{ccag}$ and $y=g\underline{ga}\overline{t}c\underline{ccag}\overline{cgt}$ be given two strings of length $12$.  Then $utd(x,y)=2$ since $x$ can be transformed into $y$ by translocating the substrings $x[3..4]=ga$ and $x[2..2]=t$, and translocating the substrings $x[6..8]=cgt$ and $x[9..12]=ccag$.
\end{example}

When $utd(x,y) <\infty$, we say that $x$ and $y$ have $utd$-match. If $x$ has $utd$-match with a suffix of $y$, we write $x \trdsuff y$.

%%%%%%%%%%%%%%%%%%%%%%%%%%%%

\section{A Dynamic Programming Solution}
\label{sec:DP}
In this section we present a general dynamic programming algorithm for the pattern matching problem with adjacent unbalanced translocations. Specifically we first describe an algorithm for solving Variant (c) of the problem. Subsequently we discuss how to slightly modify such algorithm for solving variant (a), (b) and (d).

Let $x$ be a pattern of length $m$ and $y$ a text of length $n$, with $n\geq m$, both over the same alphabet $\Sigma$ of size $\sigma$. In addition let $\delta$ be a bound for the number of translocations allowed in any approximate occurrence of $x$ in $y$, with $\delta\leq \lfloor m/2 \rfloor$.

Our algorithm is designed to iteratively compute, for $j = m,m+1,\ldots,n$, all the prefixes of $x$ which have a $\delta$-bounded $utd$-match with the suffix  of $y_{j}$, by exploiting information gathered at previous iterations.

Since the allowed edit operations involve substrings of the pattern $x$, it is useful to consider the set $\texttt{F}^k_j$ of all the positions in $x$ at which an occurrence of the suffix of $y_{j}$ of length $k$ ends.  More precisely, for $1\le k \leq m$ and $k-1\leq j \leq n$, we put 
$$ \texttt{F}^k_j := \{k-1\leq i\leq m\ |\ y[j-k+1\,..\,j] \sqsupseteq x_{i}\}\,.$$
Observe that 
\begin{equation}\label{obs1}
\texttt{F}^k_j \subseteq \texttt{F}^h_j, \textrm{ for } 1\leq h \leq k \leq m.
\end{equation}

\begin{example}\label{ex3}
Let $x=cattcatgatcat$ and $y=atcatgacttactgactta$ be a pattern and respectively
a text.  Then $\texttt{F}^3_5$ is the set of all positions in $x$ at which an occurrence of the suffix of $y_{5}$ of length $3$ ends,
namely, $cat$.  Thus $\texttt{F}^3_5 = \{3,7,13\}$.  Similarly, we have that $\texttt{F}^2_5 = \{3,7,10,13\}$.  Observe that $\texttt{F}^3_5 \subseteq \texttt{F}^2_5$.
\end{example}

The sets $\texttt{F}^k_j$ can be computed according to the following lemma.

\settowidth{\gnat}{$(k=1 \textrm{ or }i\in \mathcal{I}^{k-1}_{j-1}) \textrm{ and } p[i-k+1]=y[j]\,.$}
\settowidth{\gnatb}{$\sqcup$}
\setlength{\gnat}{-\gnat}
\addtolength{\gnat}{\textwidth}
\addtolength{\gnat}{-\gnatb}

\begin{lemma}\label{lem2}
Let $y$ and $x$ be a text of length $n$ and a pattern of length $m$, respectively.  Then $i \in \texttt{F}^k_{j}$ if and only if one of the following condition holds:
\begin{enumerate}[(a)]
\item $x[i]=y[j]$ and $k=1$;
\item $x[i]=y[j]$, $1<k<i$ and $(i-1)\in \texttt{F}^{k-1}_{j-1}$,
\end{enumerate}
for $1\leq k< i < m$ and $k-1\leq j\leq n$.
\end{lemma}

\begin{figure}[t]
\begin{center}
\setlength{\unitlength}{0.008\textwidth}
\setlength{\fboxrule}{1mm}
\begin{picture}(84,27)
 %\put(0,0){\framebox(100,21){}}

 % caso 1
 \put(0,18){\makebox(30,3)[l]{$y$}}
 \put(3,18){  $\underbrace{\framebox(21,4)[c]{...}
     		                         \framebox(14,4)[c]{$s$}
		                         \hspace{-1pt}\overbrace{\framebox(16,4)[c]{$w$}}^{k}
		                         \hspace{-1pt}\overbrace{\framebox(10,4)[c]{$z$}}^{h}
		                        }_{j}
		                        \hspace{-1pt}\framebox(21,4)[c]{...}$ 
		}

 \put(22,4){\makebox(30,3)[l]{$x$}}
 \put(24,4){ $\underbrace{
 			\framebox(14,4)[c]{$u$}
                         \hspace{-1pt}\overbrace{\framebox(10,4)[c]{$z$}}^{h}
                         \hspace{-1pt}\overbrace{\framebox(16,4)[c]{$w$}}^{k}
                       }_{i}$
                       \hspace{-3pt}\framebox(8,4)[c]{}
                 }

\end{picture}
\end{center}
\caption{\label{fg1}Case (b) of Lemma \ref{lem1}. The prefix $u$ of the pattern, of length $i-h-k$, has a $utd$-match ending at position $j-h-k$ of the text, i.e. $(i-h-k)\in \texttt{P}_{j-k-h} \cup \{0\}$. In addition the substring of the pattern $z = x[i-h-k+1..i-k]$, of length $h$ has an exact match with the substring of the text $y[j-h+1..j]$, i.e. $(i-k)\in \texttt{F}^h_j$. Finally the substring of the pattern $w = x[i-k+1..i]$, of length $k$ has an exact match with the substring of the text $y[j-h-k+1..j-h]$, i.e. $i \in \texttt{F}^k_{j-h}$.}
\end{figure}
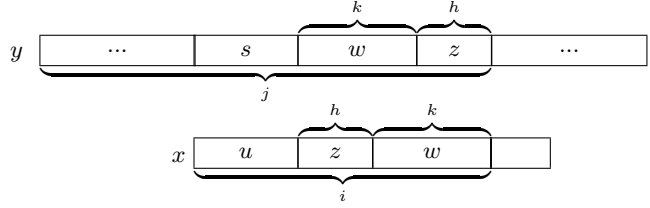

Observe that, based on equation (\ref{obs1}), we can represent the sets $\texttt{F}^{k}_{j}$ by means of a single matrix $\texttt{F}$ of size $m\times n$. Specifically, for $1\leq i\leq m$ and $1\leq j \leq n,$ we set $\texttt{F}[i,j]$ as the length of the longest suffix of $x_i$ which is also a suffix of $y_j$. More formally we have:
$$
	i \in \texttt{F}^{k}_{j} \iff \texttt{F}[i,j]\geq k.
$$

For solving the problem it is convenient to introduce the set $\texttt{Q}^i_j$, for $1\leq i\leq m$ and $1 \leq j \leq n$, defined by  
$$\texttt{Q}^i_j := \{t\ |\ t\leq\delta \textrm{ and } \ x_i \ttrdsuff y_j\}.$$
The size of the set $\texttt{Q}^i_j$ corresponds to the number of possible different alignments of the string   $x_i$ against the suffix of $y_j$ of length $i$, allowing for unbalanced translocations of adjacent factors. Observe that the size of the set $\texttt{Q}[i,j]$ is always less than $\delta$. The minimum number of translocations for transforming $x_i$ in the suffix of $y_j$ is then given by
$$
    \delta(x_i, y_[j-i+1 .. j]) = \min( \texttt{Q}^i_j \cup \{+\infty\} )
$$
where we agree upon the value $+\infty$ in the case of no possible alignment between the two strings.
As a consequence, the pattern $x$ has a $\delta$-bounded $utd$-match ending at position $j$ of the text $y$ if and only if the set $\texttt{Q}^m_j$ is not empty.

The sets $\texttt{Q}^i_j$ can then be computed by way of the recursive relations contained in the following elementary lemmas.

\begin{lemma}\label{lem1}
Let $y$ and $x$ be a text of length $n$ and a pattern of length $m$, respectively.  Then $t \in \texttt{Q}^i_j$, for $1\leq i \leq m$ and $i\leq j\leq n$,
% $P_i \trdsuff T_j$
if and only if one of the following two facts holds:
\begin{enumerate}[(i)]
    \item \label{first}
    $t\leq \delta$, $x[i]= y[j]$ and
    $t-1 \in \texttt{Q}^{i-1}_{j-1} \cup \{-1\}$;

    \item  \label{second}
    $t\leq \delta$, $(i-k)\in \texttt{F}^h_j$,  $i\in  \texttt{F}^k_{j-h}$, and $t-1 \in \texttt{Q}^{i - k-h}_{j-k-h} \cup \{-1\}$,
    for some $1 \leq k,h < i$ such that $h+k\leq i$;
\end{enumerate}
\end{lemma}

Notice that condition (\ref{second}) in Lemma~\ref{lem1} refers to a translocation of adjacent factors of length $k$ and $h$, respectively (see Figure \ref{fg1}).

For solving variant (c) of the problem the sets $\texttt{Q}^i_j$ can be maintained as a matrix $\texttt{Q}$ of integer values of size $n\times m$, where
$$
    \texttt{Q}[i,j] = \min( \texttt{Q}^i_j \cup \{+\infty\})
$$

\begin{algorithm}[ht] %\label{fig:code1}
\caption{\label{fig:code1}Dymanic programming solution}
\begin{small}
\LinesNumbered
\SetNlSty{normal}{}{.}
\SetAlgoNoEnd
\SetAlgoVlined
\KwIn{A pattern $x$, a text $y$ and a bound $\delta$}
\KwOut{The number of approximate occurrences of $x$ in $y$}
\For{$j \gets 1$ \textbf{to} $n$}{ 
    \For{$i \gets 1$ \textbf{to} $n$}{
        $\texttt{Q}[i,j] \gets \infty$\;\nllabel{algo1:init0} 
    }
}
$\texttt{Q}[0,0] \gets 0$\;\nllabel{algo1:init1} 
\For{$j \gets 1$ \textbf{to} $n$}{ 
    $\texttt{Q}[0,j] \gets 0$\;\nllabel{algo1:init2}
    \For{$i \gets 1$ \textbf{to} $m$}{ 
        \If{$x[i]=y[j]$} { 
            $\texttt{Q}[i,j]\gets \texttt{Q}[i-1,j-1]$\;\nllabel{algo1:set1} 
            $\texttt{F}[i,j] \gets \texttt{F}[i-1,j-1]+1$\;
        }
        \For{$k \gets 1$ \textbf{to} $i-1$} { 
            \For{$h \gets 1$ \textbf{to} $i-k$} { 
                \If{$\texttt{F}[i-h,j]\geq k$ \textbf{and} $\texttt{F}[i,j-k]\geq h$} { 
                    $\texttt{Q}[i,j]\gets \min(\texttt{Q}[i,j],\texttt{Q}[i-h-k,j-h-k]+1)$\;\nllabel{algo1:set2}
                }
            }
        }
    }
    \If{\emph{\texttt{Q}}$[m,j] \leq \delta$} {
        \textbf{Output} $j-m$\nllabel{algo1:output}
    }
}

\end{small}
\end{algorithm}

Based on the recurrence relations in Lemmas~\ref{lem1} and~\ref{lem2}, a general dynamic programming algorithm can be readily constructed, characterized by an overall $\bigO(n m^3)$-time and $\bigO(m^3)$-space complexity.  

The overhead due to the computation and the maintenance of the sets $\texttt{F}$ and $\texttt{Q}$ turns out to be quite relevant.

However, since we need only the last $m$ columns of the matrices $\texttt{F}$ and $\texttt{Q}$ for computing the next column, we can maintain them by means of two matrices of size $m^2$.
The pseudocode of the resulting dynamic programming algorithm is shown in Algorithm \ref{fig:code1}. Due to the four for-loops at lines 5, 7, 11, and 12, respectively, it can straightforwardly be proved that Algorithm1 has a $\bigO(n m^3)$-time and $\bigO(m^2)$-space complexity.
Indeed, the matrices $\texttt{F}$ and $\texttt{Q}$ are filled by columns and therefore we need to store only the last $m$ columns at each iteration of the for-loop at line 5.

%%%%%%%%%%%%%%%%%%%%%%%%%%%%%%%%%%

\subsection{Solving other variants of the problem}
Algorithm 1, as proposed in the previous section, solves variant (c) of the approximate string matching problem allowing for non overlapping unbalanced translocations.
In this section we briefly discuss how to slightly modify the algorithm to obtain a solution for variants (a), (b) and (d) respectively.

\smallskip 

We first start with \textbf{Variant (b)} which consists in finding the number of all $\delta$-bounded approximate occurrences of $x$ in $y$. In this case it is enough to modify Algorithm 1 in order to count the matching positions while they are given in line \ref{algo1:output}. Thus the solution maintains the same time and space complexity.

\smallskip 

\textbf{Variant (a)} consists in finding the number of all unbounded approximate occurrences of $x$ in $y$, without taking into account the approximate distances between the pattern and its occurrences in the text.  
To solve this variant of the problem it is enough to reduce the set $\texttt{Q}[i,j]$, for $1 \leq j \leq n$, to a boolean matrix where 
$\texttt{Q}[i,j] = True$ if and only if $x_i \trdsuff y_j$ and $\texttt{Q}[i,j] = False$ otherwise.
In this context, the pattern $x$ has an $utd$-match ending at position $j$ of the text $y$ if and only if $\texttt{Q}[m,j] = True$.

With reference to Algorithm \ref{fig:code1}, the inizialization of line \ref{algo1:init0} must be changed to set $\texttt{Q}[i,j] = False$, for $0\leq i \leq m$ and $0\leq j \leq n$. Similarly the initialization of line \ref{algo1:init1} and line \ref{algo1:init2} must be changed to set $\texttt{Q}[i,0] = True$, for $0\leq i \leq n$.
Finally the assignment of line \ref{algo1:set2} sets $\texttt{Q}[i,j] = \texttt{Q}[i-h-k,j-h-k]$, while the instruction of line \ref{algo1:output} is performed only when $\texttt{Q}[i,j]$ is equal to $True$.

\smallskip 

\textbf{Variant (d)} consists in finding, for each position $j$ such that $x \ttrdsuff y_j$, the number of distinct $\delta$-bounded approximate occurrences of $x$ in $y$ ending at position $j$. In this case the sets $\texttt{Q}[i,j]$, must be maintained as bit-vectors of length $\delta$, where the $t$-th bit of $\texttt{Q}[i,j]$ is set if and only if $t in \texttt{Q}[i,j]$, with $0\leq t < \delta$.

With reference to Algorithm \ref{fig:code1}, the inizialization of line \ref{algo1:init0} must be changed to set $\texttt{Q}[i,j] = 0^{\delta}$, for $0\leq i \leq m$ and $0\leq j \leq n$. Similarly the initialization of line \ref{algo1:init1} and line \ref{algo1:init2} must be changed to set $\texttt{Q}[i,0] = 10^{\delta-1}$, for $0\leq i \leq n$.
Finally the assignment of line \ref{algo1:set2} sets $\texttt{Q}[i,j] = \texttt{Q}[i-h-k,j-h-k]$. The instruction of line \ref{algo1:output} is performed only when $\texttt{Q}[i,j]\neq 0$ and returns in output all values $t$ such that the $t$-th bit of $\texttt{Q}[i,j]$ is set.

Observe that, is $\delta\leq w$, the whole algorithm requires the same time maintains the same time and space complexity of Algorithm \ref{fig:code1}, where $w$ is the size of a word in the target machine. Otherwise, when $\delta > w$ the worst case time complexity of the algorithm if $\bigO(nm^3\delta/w)$ while the space complexity is $\bigO(m^2\delta/w)$.

\section{An Automaton-Based Algorithm}\label{sec:automaton}
In this section we improve the algorithm described in Section 3 by means of an efficient method for computing the sets $\texttt{F}^{k}_{j}$, for $1 \leq j \leq n$ and $1 \leq k < j$. Such method makes use of the \textsc{Dawg} of the pattern $x$ and the function $\ep$. In addition we introduce an efficient method for maintaining information from the sets $\texttt{Q}_i^j$.

\smallskip

Let $\mathcal{A}(x)=(Q, \Sigma, \delta, \myroot, F)$ be the \textsc{Dawg} of $x$.  For each position $j$ in $y$, let $w$ be the longest factor of
$x$ which is a suffix of $y_j$ too; also, let $q_j$ be the state of $\mathcal{A}(x)$ such that $\R(w)=q_j$, and let $l_j$ be the length of $w$.  We call the pair $(q_j,l_j)$ a
\emph{$y$-configuration} of $\mathcal{A}(x)$.  

The idea is to compute the $y$-configuration $(q_j, l_j)$ of $\mathcal{A}(x)$, for each position $j$ of the text, while scanning the text $y$.  The set $\texttt{F}^{k}_{j}$ computed at previous iterations do not need to be maintained explicitly; rather, it is enough to maintain only $y$-configurations.  These are then used to compute efficiently the set $\texttt{F}^{k}_{j}$ only when needed.

\begin{example}\label{ex4}
Let $x=aggga$ be a pattern and $y=aggagcatgggactaga$ a text respectively. Let $\mathcal{A}(x)=(Q, \Sigma, \delta, \myroot, F)$ be the DAWG of $x$ as depicted in Figure \ref{fig:dawg}, where $\myroot=q_0$ is the initial state and $F=\{q_1,q_2,q_3,q_4,q_5,q_6,q_7\}$ is the set of final states.  Edges of the \textsc{Dawg} are depicted in black while suffix links are depicted in red.
Observe that, after scanning the suffix $y_5$ starting from state $q_0$ of $\mathcal{A}(x)$, we reach state $q_2$ of the automaton. Thus, the corresponding $y$-configuration is $(q_2,2)$. Similarly, after scanning the suffix $y_{11}$, we get the $y$-configuration $(q_4,3)$.
\end{example}

%%%%%%%%%%%%%%%%%%%%%%%%%%%%%%%%%%%%%%
% Computation of the longest factor  %
%%%%%%%%%%%%%%%%%%%%%%%%%%%%%%%%%%%%%%

The longest factor of $x$ ending at position $j$ of $y$ is computed in the same way as in the Forward-Dawg-Matching algorithm for the exact pattern matching problem (the interested readers are referred to \cite{CR94} for further details).

Specifically, let $(q_{j-1},l_{j-1})$ be the $y$-configuration of $\mathcal{A}(x)$ at step $(j-1)$.
The new $y$-configuration $(q_j, l_j)$ is set to $(\delta(q,y[j]), length(q)+1)$, where $q$ is the first node in the suffix path $\langle q_{j-1}$, $\suf(q_{j-1})$, $\suf^{(2)}(q_{j-1}),\ldots\rangle$ of $q_{j-1}$, including $q_{j-1}$, having a transition on $y[j]$, if such a node exists; otherwise $(q_j, l_j)$ is set to $(\myroot,0)$.\footnote{We recall that $\suf^{(0)}(q) \defAs q$ and,  recursively, $\suf^{(h+1)}(q) \defAs \suf(\suf^{(h)}(q))$, for  $h\geq 0$. %provided that $\suf^{(h)}(q) \neq \myroot$.
}

\smallskip

Before explaining how to compute the sets $\texttt{F}^{k}_{j}$, it is convenient to introduce the partial function $\phi:Q \times \mathbb{N} \rightarrow Q$, which, given a node
$q \in Q$ and a length $k \leq length(q)$, computes the state $\phi(q,k)$ whose corresponding set of factors contains the suffix of $\val(q)$ of length $k$.  
Roughly speaking, $\phi(q,k)$ is the first node $p$ in the suffix path of $q$ such that $length(\suf(p)) < k$.

In the preprocessing phase, the \textsc{Dawg} $\mathcal{A}(x)=(Q, \Sigma,\delta, \myroot, F)$, together with the associated $\ep$ function, is computed.  Since for a pattern $x$ of length $m$ we have $|Q| \leq 2m+1$ and $|\ep(q)| \leq m$, for each $q \in Q$, we need only $\bigO(m^2)$ extra space (see \cite{CR94}).

To compute the set $\texttt{F}^k_j$, for $1\le k\le l_j$, we take advantage of the relation
\begin{equation}
    \label{eq:fac}
    \texttt{F}^k_j=\ep(\phi(q_j,k)).
\end{equation}
Notice that, in particular, we have $\texttt{F}^{l_j}_j=\ep(q_j)$.

The time complexity of the computation of $\phi(q,k)$ can be bounded by the length of the suffix path of node $q$.  Specifically, since the sequence
$$\langle\length(\suf^{(0)}(q)),\length(\suf^{(1)}(q)), \ldots, 0\rangle$$
of the lengths of the nodes in the suffix path from $q$ is strictly decreasing, we can do at most $\length(q)$ iterations over the suffix link, obtaining a $\bigO(m)$-time complexity.

%%%%%%%%%%%%%%%%%%%%%%%%%%%%%%%%%%%%%%
% Computation of sets F^k_j  	     %
%%%%%%%%%%%%%%%%%%%%%%%%%%%%%%%%%%%%%%
According to Lemma~\ref{lem1}, a translocation of two adjacent factors of length $k$ and $h$, respectively, at position $j$ of the text $y$ is possible only if two factors of $x$ of lengths at least $k$ and $h$, respectively, have been recognized at both positions $j$ and $j-h$, namely if $l_j\ge h$ and $l_{j-h}\ge k$ hold (see Figure \ref{fg1}).

Let $\langle h_1, h_2, \ldots, h_r \rangle$ be the increasing sequence of all the values $h$ such that $1\leq h \leq \min(l_{j},l_{j-h})$.  For each $1\leq i \leq r$, condition (\ref{second}) of Lemma~\ref{lem1} requires member queries on the sets $\texttt{F}^{h_i}_{j}$.

Observe that if we proceed for decreasing values of $h$, the sets $\texttt{F}^{h}_{j}$, for $1 \leq h \leq l_j$, can be computed in constant time.  Specifically, for $h=1, \ldots, l_j-1$ $\texttt{F}^h_j$ can be computed in constant time from $\texttt{F}^{h+1}_j$ with at most one iteration over the suffix link of the state $\phi(q_{j},h+1)$.

Subsequently, for each member query on the set $\texttt{F}^{h_i}_{j}$, condition (\ref{second}) of Lemma~\ref{lem1} requires also member queries on the sets $\texttt{F}^{k}_{j-{h_i}}$, for $1\leq k < h_i$. Let $\langle k_1, k_2, \ldots, k_s \rangle$ be the increasing sequence of all the values $k$ such that $1\leq k \leq \min(l_{j-h_i},l_{j-h_i-k})$.
Also in this case we can proceed for decreasing values of $k$, in order to compute the sets $\texttt{F}^{k}_{j-{h_i}}$ in constant time, for $1 \leq k \leq l_{j-h_i}$. 

\smallskip

In order to efficiently maintain the elements in the sets $\texttt{Q}_i^j$ we introduce the sets $\texttt{P}_j$, for $1\leq j \leq n$, as the set of all prefixes of the pattern $x$ which have an approximate occurrences ending at position $j$ of the text. More formally
$$
    \texttt{P}_j = \{1\leq i \leq m\ |\ \texttt{Q}_i^j]\leq \delta \}.
$$
Such sets can be maintained as linked lists in order to be able to scan in linear time, for each position $j$ of the text, the set of all position $i$ for which it holds that $\texttt{Q}_i^j\leq \delta$.

\smallskip

The resulting algorithm for the approximate string matching problem allowing for unbalanced translocations of adjacent factors is shown in Algorithm \ref{fig:code2}. In the next sections, we analyze its worst-case and average-case complexity.

\begin{figure}
\begin{center}
\includegraphics[width=0.6\textwidth]{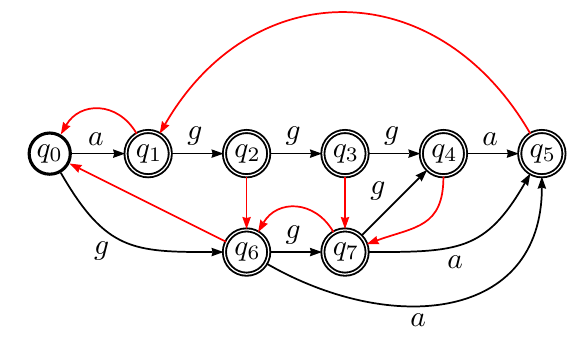}
\end{center}
\caption{\label{fig:dawg}The Directed Acyclic Word Graph (\textsc{Dawg}) of the pattern $x=aggga$, where the edges of the automaton are depicted in black while the suffix links are depicted in red.}
\end{figure}

\begin{algorithm}[ht] %\label{fig:code1}
\caption{\label{fig:code2}The automaton based solution of variant (c)}
\begin{small}
\LinesNumbered
\SetNlSty{normal}{}{.}
\SetAlgoNoEnd
\KwIn{A pattern $x$ of length $m$, a text $y$ of length $n$ and a bound $\delta$}
\KwOut{The number of approximate occurrences of $x$ in $y$}
%$m \gets \slength{x}$\;
%$n \gets \slength{y}$\;
\For{$j \gets 1$ \textbf{to} $n$}{ 
    \For{$i \gets 1$ \textbf{to} $n$}{
        $\texttt{Q}[i,j] \gets \infty$\;\nllabel{algo2:init0} 
    }
}
$(q_0,l_0) \gets$ Dawg-Delta$(root_{\aut},0,y[0],\aut)$\;
$\texttt{Q}[0,0] \gets 0$\;
$\texttt{P}_0 \gets \{0\}$\;
\nllabel{l:main}\For{$j \gets 1$ \textbf{to} $n$} {
    $(q_j,l_j) \gets$ Dawg-Delta$(q_{j-1},l_{j-1},y[j],\aut)$\;
    $\texttt{P}_j \gets \{0\}$\;
    $\texttt{Q}[0,j] \gets 0$\;
%    \nllabel{loop:2}\For{$i \gets 1$ \textbf{to} $m$}{ 
%        \If{$x[i]=y[j]$} { 
%            $\texttt{Q}[i,j]\gets \texttt{Q}[i-1,j-1]$\;\nllabel{algo1:set1} 
%        }
%    }
    \nllabel{loop:2}\For{$i \in$ \emph{\texttt{P}}$_{j-1}$} {
        \If{$i<m$ \textbf{and} $x[i+1]=y[j]$} {
            $\texttt{Q}[i+1,j]\gets \texttt{Q}[i,j-1]$\;
        }
    }
    $u\gets q_j$\;
    \nllabel{loop:3}\For{$h\gets l_j$ \textbf{downto} $1$} {
        \If {$h = \length(\suf_{\aut}(u))$} {
            $u\gets \suf_{\aut}(u)$\;
        }
        $p\gets q_{j-h}$\;
        \nllabel{loop:4}\For{$k\gets l_{j-h}$ \textbf{downto} $1$} {
            \If {$k = \length(\suf_{\aut}(p))$} {
                $p\gets \suf_{\aut}(p)$\;
            }
            \nllabel{loop:5}\For{$i\in$ \emph{\texttt{P}}$_{j-h-k}$} {
                \nllabel{l:end}\If{$(i+h)\in \ep(u)$ \textbf{and} $(i+h+k)\in \ep(p)$} {
                    \If{\emph{\texttt{Q}}$[i,j] + 1 \leq \delta$} {
                        $\texttt{Q}[i+h+k,j] \gets \min(\texttt{Q}[i+h+k,j],\texttt{Q}[i,j] + 1)$\;
                        $\texttt{P}_j\gets \texttt{P}_j\cup \{(i+h+k)\}$\;
                    }
                }
            }
        }
    }
    \If {\emph{\texttt{Q}}$[m,j] \leq \delta$} {
        \textbf{Output} $j-m$\;
    }
}
\end{small}
\end{algorithm}

\begin{procedure}[h!] 
\caption{Dawg-Delta($q, l, \mathcal{A}$)}
\begin{small}
\LinesNumbered
\SetNlSty{normal}{}{.}
\SetAlgoNoEnd
\KwIn{The current configuration $(q,l)$ of the automaton $\mathcal{A}$ and a character $c$}
\KwOut{The new configuration of the automaton}
\uIf{$\delta_{\mathcal{B}}(q, c) = \textsc{nil}$} {
    \Repeat{$q = \textsc{nil}$ \textbf{or} $\delta_{\mathcal{A}}(q, c) \neq \textsc{nil}$} {
        $q\gets \isuf_{\mathcal{A}}(q)$\;
    }
    \uIf{$q = \textsc{nil}$} {
        $l\gets 0$\;
        $q\gets \myroot_{\mathcal{A}}$\;
    }
    \Else {
        $l\gets \length(q)+1$\;
    }
    $q\gets \delta_{\mathcal{A}}(q,c)$\;
}
\Else {
    $l\gets l+1$\;
    $q\gets \delta_{\mathcal{A}}(q, c)$\;
}
\Return $(q,l)$\;
\end{small}
\label{fig:dawgdelta}
\end{procedure}

%%%%%%%%%%%%%%%%%%%%%%%%%%%%%%%%%%%%%%%%%%%%%
%        Complessita'               %
%%%%%%%%%%%%%%%%%%%%%%%%%%%%%%%%%%%%%%%%%%%%%

\subsection{Worst-case time and space analysis}

In this section we determine the worst-case time and space analysis of the algorithm  presented in the previous section.  In particular, we will refer to the solution reported in Algorithm~\ref{fig:code2}.

First of all, observe that the main for-loop at line 7 is always executed $n$ times.  For each of its iterations, the cost of the execution of Procedure Dawg-Delta (line 8) for computing the new $y$-configuration requires at most $\bigO(m)$-time.  Since we have $|\texttt{P}_j|\leq m+1$, for all $1\leq j \leq n$, the for-loop at line 11 is also executed $\bigO(m)$-times. In addition, since we have $l_j\leq m$, for all $1\leq j \leq n$, the two nested for-loops at lines 15 and 19 are executed $m$ times. Observe also that the transitions of suffix links performed at lines 17 and 20 need only constant time. Thus, at each iteration of the main for-loop, the internal for-loop at line 19 takes at most $\bigO(m)$-time, while the for-loop at line 15 takes at most $\bigO(m^2)$-time.  In addition the for-loop at line 22 takes at most $\bigO(m)$-time, since $|\texttt{P}_{j-h-k}|\leq m+1$ and the tests at line 23 can be performed in constant time.
Summing up, Algorithm \ref{fig:code2} has a $\bigO(nm^3)$ worst-case time complexity.

%%%%%%%%%%%%%%%%%%%%%%%%%%%%%%%%%%%%%%%%%%%%%
%        Complessita' spaziale              %
%%%%%%%%%%%%%%%%%%%%%%%%%%%%%%%%%%%%%%%%%%%%%
%\subsection{Space complexity}
In order to evaluate the space complexity of Algorithm \ref{fig:code2}, we observe that in the worst case, during the $j$-th iteration of its main for-loop, the sets $\texttt{F}^k_{j-k}$ and $\texttt{P}_{j-k}$, for $1\leq k \leq m$, must be kept in memory to handle translocations.
However, as explained before, we do not keep the values of
$\texttt{F}^k_{j-k}$ explicitly but rather we maintain only their corresponding $y$-configurations of the automaton $\mathcal{A}(x)$.
Thus, we need $\bigO(m)$-space for the last $m$ configurations of the automaton and $\bigO(m^2)$-space to keep the last $m+1$ values of the sets $\texttt{P}_{j-k}$, since the maximum cardinality of each such set is $m+1$.  Observe also that, although the size of the \textsc{Dawg} is linear in $m$, the \emph{end\textendash pos}$(\cdot)$ function can require $\bigO(m^2)$-space.  Therefore, the total space complexity of Algorithm \ref{fig:code2} is $\bigO(m^2)$.

%%%%%%%%%%%%%%%%%%%%%%%%%%%%%%%%%%%%%%%%%%%%%
%        Complessita' nel caso medio        %
%%%%%%%%%%%%%%%%%%%%%%%%%%%%%%%%%%%%%%%%%%%%%
\subsection{Average-case time analysis}
\label{av_case}

In this section we evaluate the average time complexity of our new automaton-based algorithm, assuming the uniform distribution and independence of characters in an alphabet $\Sigma$ with $\sigma\geq 4$ characters. 

In our analysis we do not include the time required for the computation of the DAWG and the \emph{end\textendash pos}$(\cdot)$ function, since they require $\bigO(m)$ and $\bigO(m^2)$ worst-case time, respectively, which turn out to be negligible if we assume that $m$ is much smaller than $n$. Hence we evaluate only the searching phase of the algorithm. 

\newcommand{\Xj}[1]{\X(#1)}
\newcommand{\Yj}[1]{\Y(#1)}
\newcommand{\Zj}[1]{\Z(#1)}

Given an alphabet
$\Sigma$ of size $\sigma \geq 4$, for $j=1,\ldots,n$, we consider
the following nonnegative random variables over the sample space of
the pairs of strings $x,y \in \Sigma^{*}$ of length $m$ and $n$,
respectively:
\begin{itemize}
\item $\Xj{j} \defAs $ the length $l_j\le m$ of the longest factor of $x$ which is a 	suffix of $y_{j}$;
\item $\Zj{j} \defAs |\texttt{P}_j|$, where we recall that $\texttt{P}_j = \{1\leq i\leq m\ |\ x_i \trdsuff y_j\}$.
\end{itemize}
Then the run-time of a call to Algorithm \ref{fig:code2} with parameters $(x,y)$ is proportional to 
\begin{equation} \label{sums}
\sum_{j=1}^{n}\left( \Zj{j-1} + \Xj{j} + \sum_{h=1}^{\Xj{j}} \left( \Xj{j-h} +  \sum_{k=1}^{\Xj{j-h}} \Zj{j-h-k} \right)\!\right),
\end{equation}
where the external summation refers to the main for-loop (at line 7), the second summation within it takes care of the internal for-loop at line 15, and the third summation refers to the inner for-loop at line 19.

Hence the average-case complexity of Algorithm \ref{fig:code2} is the expectation of (\ref{sums}), which, by linearity, is equal to

\begin{footnotesize}
\begin{equation}\label{sumsb}
\sum_{j=1}^{n}\!\!\left(\! \expect(\Zj{j-1}) + \expect( \Xj{j}) + \expect\left(\! \sum_{h=1}^{\Xj{j}} \Xj{j-h} \!\right)+ \expect\left(\! \sum_{h=1}^{\Xj{j}}\sum_{k=1}^{\Xj{j-h}} \Zj{j-h-k}\!\right)  \!\right) \!\!.
\end{equation}
\end{footnotesize}

where $\expect(\cdot)$ be the expectation function. Since $\expect(\Xj{j}) \leq \expect(\Xj{n-1})$ and $\expect(\Zj{j}) \leq \expect(\Zj{n-1})$, for $1 \leq j \leq n$,\footnote{In fact, for $j=m+1,\ldots,n$ all the inequalities hold as equalities.} by putting $\X  \defAs  \Xj{n-1}$ and $\Z  \defAs  \Zj{n-1}$, the expression (\ref{sumsb}) gets bounded from above by
\begin{equation}
    \label{sumsc}
\sum_{j=1}^{n}\!\!\left( \expect(\Z) + \expect( \X) + \expect  \left( \sum_{h=1}^{\X}  \X \right) + \expect  \left( \sum_{h=1}^{\X}\sum_{k=1}^{\X} \Z \right) \right) .
\end{equation}

Let $Z_{i}$ and $X_h$ be the indicator variables defined for $i=1,\ldots,m$ and $h=1,\ldots,m$, respectively as
$$
\begin{array}{ll}
Z_{i} \defAs \begin{cases}
1 & \text{if } i \in \texttt{P}_{n}\\
0 & \text{otherwise}
\end{cases}
& \textrm{ and~~ }
X_{h} \defAs \begin{cases}
1 & \text{if } X \geq h\\
0 & \text{otherwise}\,,
\end{cases}
\\
\end{array}
$$
Hence
$$
\begin{array}{l}
\displaystyle \Z = \sum_{i=1}^{m}\Z_{i},\ \expect(\Z_{i}^{2}) = \expect(\Z_{i}) = \Prob\{x_{i} \trdsuff y\},\\
\displaystyle \X = \sum_{h=1}^{m}\X_{h},\  \textrm{ and }\ \expect(\X_{h}^{2}) = \expect(\X_{h}) = \Prob\{X \geq h\}.\\
\end{array}
$$
So that we have
$$
\sum_{h=1}^{\X} \X = \X\X =
    \left(\sum_{h=1}^{m} \X_{h}\right) \cdot \left(\sum_{k=1}^{m}\X_{k}\right) =   \sum_{h=1}^{m}\sum_{k=1}^{m}\X_{h}\X_{k}\,.
$$
Therefore
\begin{equation*}
%     \label{star1}
\sum_{h=1}^\X \sum_{k=1}^\X \Z = \X\X\Z = \left( \sum_{h=1}^m \sum_{k=1}^{m} \X_h \X_k \right) \cdot \Z,
\end{equation*}
which yields the following upper bound for (\ref{sumsc}):
\begin{equation}
    \label{sumsd}
\sum_{j=1}^{n}\!\!\left( \expect(\Z) + \expect( \X) + \left( \expect(\Z) + 1\right) \cdot \sum_{h=1}^m \sum_{k=1}^{m}\expect\!\left(\X_h \X_k\right)  \right).
\end{equation}

To estimate each of the terms $\expect(\X_h\X_k)$ in
(\ref{sumsd}), we use the well-known Cauchy-Schwarz inequality which
in the context of expectations assumes the form
    $|\expect(UV)| \le \sqrt{\expect(U^{2})\expect(V^{2})}$\,,
for any two random variables $U$ and $V$
such that $\expect(U^{2})$, $\expect(V^{2})$ and $\expect(UV)$ are
all finite.

Then, for $1 \leq h \leq m$ and $1\leq k \leq m$, we have
\begin{equation}
    \label{expProduct}
    \expect(X_hX_k) \le \sqrt{\expect(X_h^{2})\expect(X_k^{2})}
    =\sqrt{\expect(X_h)\expect(X_k)}\,.
\end{equation}

From (\ref{expProduct}), it then follows that (\ref{sumsd}) is bounded
from above by
\begin{small}
%$$
\begin{eqnarray}
    \label{bound}
     \nonumber\displaystyle\sum_{j=1}^{n}\left( \expect(\Z) + \expect(\X) + \left( \expect(\Z) + 1\right) \cdot \sum_{h=1}^m     \sum_{k=1}^{m}     \sqrt{\expect(X_h) \expect(X_k)} \right)\\
     \displaystyle \qquad= \sum_{j=1}^{n}\!\left( \expect(\Z)    +    \expect( \X) + \left( \expect(\Z) + 1\right) \cdot      \left(\sum_{h=1}^m\sqrt{\expect(X_h)}\right) \cdot    \left(\sum_{k=1}^{m}\sqrt{\expect(X_k)}\right) \right)\\
    \nonumber \displaystyle \qquad= \sum_{j=1}^{n}\!\left( \expect(\Z)    +    \expect( \X) + \left( \expect(\Z) + 1\right) \cdot      \left(\sum_{h=1}^m\sqrt{\expect(X_h)}\right)^2  \right).
\end{eqnarray}
%$$
\end{small}

To better understand (\ref{bound}), we evaluate the
expectations $\expect(\X)$ and $\expect(\Z)$ and the sum
$\sum_{h=1}^m\sqrt{\expect(X_h)}$.  To this purpose, it will be
useful to estimate also the expectations
%\begin{itemize}
    $\expect(X_h)=\Prob\{X\ge h\}$, for $1 \leq h \leq m$, and
    $\expect(X_k)=\Prob\{x_i \trdsuff y\}$, for $1\leq k \leq  m$.
%\end{itemize}

Concerning $\expect(X_h)=\Prob\{X\ge h\}$, we reason as follows.
Since $y[n-h+1\,..\,n]$ ranges uniformly over a collection of
$\sigma^{h}$ strings and there can be at most $\min(\sigma^h, m-h+1)$
distinct factors of length $h$ in $x$, the probability $\Prob\{X\ge
h\}$ that one of them matches $y[n-h+1\,..\,n]$ is at most
$\min\left(1,\frac{m-h+1}{\sigma^h}\right)$. Hence, for  $h =
1,\ldots,m$, we have
\begin{equation}
    \label{star3}
\expect(X_h) \leq \min\left(1,\frac{m-h+1}{\sigma^h}\right).
\end{equation}

In view of (\ref{star3}),  we have:
\begin{equation}
    \label{star4}
\expect(\X)=\sum_{i=0}^m i\cdot\Prob\{\X = i\}=\sum_{i=1}^m
\Prob\{\X \ge i\} \leq \sum_{i=1}^m
\min\left(1,\frac{m-i+1}{\sigma^i}\right).
\end{equation}
Let $\overline{h}$ be the smallest integer $1\le h< m$ such that
$\frac{m-h+1}{\sigma^h} < 1$. Then, from (\ref{star4}), we have
\begin{align}
    \label{star4bis}
\expect(\X) &\leq \sum_{i=1}^{\overline{h}-1} 1 +
\sum_{i=\overline{h}}^m \frac{m-i+1}{\sigma^i} \leq \overline{h} - 1 +
(m -\overline{h} +1)\sum_{i=\overline{h}}^m \frac{1}{\sigma^i} \notag\\
& < \overline{h} - 1 + \frac{\sigma}{\sigma-1}\cdot\frac{m
-\overline{h} +1}{\sigma^{\overline{h}}} < \overline{h} - 1 +
\frac{\sigma}{\sigma-1} < \overline{h} + 1\,.
\end{align}
Since $\frac{m -(\overline{h}+1) +1}{\sigma^{\overline{h}+1}} \geq 1$,
then $\sigma^{\overline{h}+1} \leq m -(\overline{h}+1) +1 \leq m-1$,
so that
$\overline{h}+1 < \log_{\sigma} m$.
Hence from (\ref{star4bis}) and $\overline{h}+1 < \log_{\sigma} m$, we obtain
\begin{equation}
    \label{Xbound}
    \expect(\X) < \log_{\sigma} m\,.
\end{equation}
Likewise, from (\ref{star3}) and $\overline{h}+1 < \log_{\sigma} m$, we have
\begin{align}
\sum_{h=1}^m&\sqrt{\expect(X_h)}  \le
\sum_{h=1}^m\sqrt{\min\left(1, \frac{m-h+1}{\sigma^k}\right)} %\\
= \sum_{h=1}^{\overline{h}-1}1 + \sum_{h=\overline{h}}^m
\sqrt{\frac{m-h+1}{\sigma^h}} \notag\\
&~~~~~~\le \overline{h}-1 +
\sqrt{m-\overline{h}+1}\cdot\sum_{h=\overline{h}}^m\frac{1}{\sqrt{\sigma^h}}%\\
< \overline{h}-1 +
\frac{\sqrt{\sigma}}{\sqrt{\sigma}-1}\cdot\sqrt{\frac{m-\overline{h}+1}{\sigma^{\overline{h}}}}
\notag\\
\label{Xbound2}
&~~~~~~< \overline{h}-1 + \frac{\sqrt{\sigma}}{\sqrt{\sigma}-1} \leq
\overline{h} + 1 < \log_{\sigma} m\,,
\end{align}
where $\overline{h}$ is defined as above.

Next we estimate $\expect(Z_i)=\Prob\{x_i \trdsuff y\}$, for $1\leq i \leq m$.
Let us denote by $\md(i)$ the number of distinct strings which have a
$utd$-match with a given string of length $i$ and whose characters are
pairwise distinct.  Then
$
\Prob\{x_{i} \trdsuff y\} \leq \md(i+1)/\sigma^{i+1}
$.
From the recursion
$$
\left\{
\begin{array}{rcl}
    \displaystyle\md(0) & = & 1\\
    \displaystyle\md(k+1) & = & \displaystyle\sum_{h=0}^{k} \md(h) + \sum_{h=1}^{\lfloor    \frac{k-1}{2}\rfloor} \md(k-2h-1) \qquad \text{(for  $k \geq 0$)}\,,
\end{array}
\right.
$$
it is not hard to see that $\md(i+1) \leq 3^{i}$, for
$i=1,2,\ldots,m$, so that we have
\begin{equation}
    \label{star6}
    \expect(\Z_{i}) = \Prob\{x_{i} \trdsuff y\} \leq    \frac{3^{i}}{\sigma^{i+1}}\,.
\end{equation}

Then, concerning $\expect(\Z)$, from (\ref{star6}) we have
\begin{equation}
    \label{Zbound}
\expect(\Z) = \expect\!\left(\sum_{i=1}^{m}\Z_{i}\right) =
\sum_{i=1}^{m}\expect(\Z_{i})
\leq \sum_{i=1}^{m}\frac{3^{i}}{\sigma^{i+1}}
< \frac{1}{\sigma} \cdot \frac{1}{1 - \frac{3}{\sigma}}
= \frac{1}{\sigma - 3} \leq 1
\end{equation}
(we recall that we have assumed $\sigma \geq 4$).

From (\ref{Zbound}), (\ref{Xbound}), and (\ref{Xbound2}), it then follows that (\ref{sumsb}) is bounded from above by $n \cdot (1+\log_{\sigma}m + 3 \log^2_{\sigma}m)$,
yielding a $\bigO(n\log^2_{\sigma}m)$ average-time complexity for our automaton based algorithm.

%%%%%%%%%%%%%%%%%%%%%%%%%%%%

%%%%%%%%%%%%%%%%%%%%%%%%%%%%%%%%%%%%%%%%%%%%
%%%%%%%%%%%%%%%%%%%%%%%%%%%%%%%%%%%%%%%%%%%%
%%%%%%%%%%%%%%%%%%%%%%%%%%%%%%%%%%%%%%%%%%%%

\section{An Alignment Based Solution} \label{sec:heuristics}
In this section we present a third solution for the approximate string alignment problem allowing for unbalanced translocations of adjacent factors. 
This solution is based on a procedure used for checking whenever an approximate match exists between two equal length  strings $x$ and $z$. 

The corresponding approximate string matching algorithm allowing for unbalanced translocations of adjacent factors can be trivially obtained by iterating the given procedure for all possible subsequences of the text of length $m$. 

The algorithm is composed by a preprocessing and searching phase, which we describe in Section \ref{sec:preprocessing} and in Section \ref{sec:searching}, respectively.
Then, in Section \ref{sec:complexity}, we prove the correctness of the algorithm and discuss its worst case time complexity.

\subsection{The Preprocessing Phase} \label{sec:preprocessing}
During the preprocessing phase of the algorithm three functions are computed, in the form of tables, which are then used during the alignment process.

%%% NEXT POSITION FUNCTION
We first define the \emph{next position function} $\mu_{x}:\Sigma \times \{1, \ldots, m\} \rightarrow \{2, \ldots, m\}$, associated to a given pattern $x$ of length $m$, as the function which returns the next position (to a given input position $i$) where a given character $c \in \Sigma$ occurs.
Specifically $\mu_{x}(c,i)$ is defined as the position $j>i$ in the pattern such that $x[j]=c$. If such a position does not exist then we set $\mu(c.i)=m+1$. More formally
$$
	\mu_x(c,i) := \min\Big(\{ j\ |\ 1\leq i < j \leq m \textrm{ and } x[j]=c\} \cup \{m+1\}\Big)
$$
The next position function $\mu_x$ can be precomputed and maintained in a table of size $m\times \sigma$ in $\bigO(m\sigma+m^2)$ time by using Procedure Compute-Next-Position show below.

\begin{procedure}[h!] \label{proc:cnp}
\caption{Compute-Next-Position($x$)}
\begin{small}
\LinesNumbered
\SetNlSty{normal}{}{.}
\SetAlgoNoEnd
\KwIn{A string $x$ of length $m$}
\KwOut{A table representing the next-position function $\mu_x$}
\ForEach{$c \in \Sigma$} { 
    \For{$i\gets 1$ \textbf{to} $m$} { 
        $\mu(c,i) \gets m+1$\;
    }
}
\For{$i\gets m$ \textbf{downto} $2$} { 
    \For{$j\gets i-1$ \textbf{downto} $1$} { 
        $\mu(x[i],j)\gets i$\;
    }
}
\Return $\mu$\; 
\end{small}
\end{procedure}

\begin{example}\label{ex13}
Let $x=gtgtaccgtgt$ be a string of length $m=11$. We have $\mu_x(g,1) = 3$, $\mu_x(g,4)=\mu_x(g,5)=8$, $\mu_x(g,8)=10$, while $\mu_x(g,10)=12$. 
\end{example}

%%% BORDER SET FUNCTION
We also define the \emph{border set function} $\psi_x$ of a given string $x$ as the set of the lengths of all borders of $x$. Specifically we define $\psi_x(i,j)$, for each $1\leq i<j \leq m$,  as the set of the lengths of all borders of the string $x[i .. j]$, so that $k\in \psi_x(i,j)$ if and only if the string $x[i .. j]$ has a border of length $k$. Formally we have
$$
	\psi_x(i,j) := \{ k \ |\ 0<k<j-i \textrm{ and }  x[i .. i+k-1] = x[j-k+1 .. j] \}
$$
\begin{example}\label{ex14}
Let $x=gtgtaccgtgt$ be a string of length $m=11$. We have 
$$
\begin{array}{ll}
 \psi_x(1,11)= \{2,4,11\},  & \textrm{ since the set of borders of $gtgtaccgtgt$ is } \{gt, gtgt, gtgtaccgtgt\};\\
 \psi_x(1,4)= \{2,4\}, & \textrm{ since the set of borders of $gtgt$ is } \{gt, gtgt\};\\
 \psi_x(4,9)= \{1,6\}, & \textrm{ since the set of borders of $taccgt$ is } \{t, taccgt\}; \\
 \psi_x(5,7)= \{3\} & \textrm{ since the set of borders of $acc$ is } \{acc\}.\\
\end{array}
$$
\end{example}

\begin{procedure}[h!] \label{proc:cbs}
\caption{Compute-Border-Set($x$)}
\begin{small}
\LinesNumbered
\SetNlSty{normal}{}{.}
\SetAlgoNoEnd
\KwIn{A string $x$ of length $m$}
\KwOut{A table representing the border-set function $\Psi_x$}
\For{$i\gets 1$ \textbf{to} $m$} { 
    \For{$j\gets i$ \textbf{to} $m$} { 
        \For{$k\gets i$ to $j$} { 
            $\Psi[i,j,k] \gets 0$\;
        }
    }
}
\For{$i\gets 1$ to $m$} { 
    \For{$j\gets i$ to $m$} {
        $\pi\gets$ Compute-Border-Table$(x,i,j)$\; 
        \For{$k\gets 0$ \textbf{to} $j-i+1$} { 
            $\Psi[i,j,\pi[k]]\gets 1$ \; 
        }
    }
}
\Return $\Psi$\; 
\end{small}
\end{procedure}

For each $i,j$ with $1\leq i \leq j \leq m$, we can represent the set $\psi_x(i,j)$ by using a vector of $(j-i+1)$ boolean values such that its $k$-th entry is set iff $k \in \psi_x(i,j)$. More formally the function $\psi_x$ can be maintained using a tridimensional bit-table $\Psi_x$, which we call \emph{border set table of $x$}, defined as
$$
	\Psi_x[i,j,k] := \left\{
			\begin{array}{ll}
				1 ~~~~	& \textrm{ if } x[i .. i+k-1] = x[j-k+1 .. j]\\
				0 		& \textrm{ otherwise } 
			\end{array}
	\right.
$$
for $1\leq i<j\leq m$ and $k < j-i$.

The border set table $\Psi_x$ can be computed in $\bigO(m^3)$-time and space by using Procedure Compute-Border-Set, where Compute-Border-Table is the $O(m)$ function used in the Knuth-Morris-Pratt algorithm \cite{KMP77}.

Observe that using $\Psi_x$ we can answer in constant-time to queries of the the type ``\emph{is $k$ the length of a border of the substring $x[i .. j]$?}'', which translates to evaluate if $\Psi[i,j,k]$ is set.

%%% SHORTEST BORDER FUNCTION

In addition we define the \emph{shortest border function} of a string $x$, as the function $\rho_x: \{1, \ldots, m\} \times \{1, \ldots, m\} \rightarrow \{1, \ldots, m\}$ which associates any nonempty substring of $x$ to the length of its shortest border. Specifically we set $\rho_x(i,j)$ to be the length of the shortest border of the string $x[i \ldots j]$, for $1\leq i <j \leq m$. 
More formally we have
$$
	\rho_x(i,j) := \min  \{ k\ |\ 0\leq k<j-i \textrm{ and } x[i .. i+k-1] = x[j-k+1 .. j] \}   = \min(\psi(i,j))
$$
It is trivial to observe that, if we already computed the border set function $\psi_x$, for the pattern $x$, the shortest border function $\rho_x$ of $x$ can be computed in $\bigO(m^3)$-time using $\bigO(m^2)$ space.

\begin{example}\label{ex5}
Let $x=gtgtaccgtgt$ be a string of length $m=11$. According to Example \ref{ex14}, we have 
$$
\begin{array}{ll}
 \rho_x(1,11)= 2,  & \textrm{ since $gt$ is the shortest nonempty border of $gtgtaccgtgt$};\\
 \rho_x(1,4)= 2, & \textrm{ since $gt$ is the shortest nonempty border of $gtgt$};\\
 \rho_x(4,9)= 1, & \textrm{ since $t$ is the shortest nonempty border of $taccgt$}; \\
 \rho_x(5,7)= 3 & \textrm{ since $acc$ is the shortest nonempty border of $acc$}.\\
\end{array}
$$
\end{example}

In what follows we will use the symbols $\mu$, $\psi$ and $\rho$, in place of  $\mu_x$, $\psi_x$ and $\rho_x$, respectively, when the reference to $x$ is clear from the context.

%%%%%%%%%%%%%%%%%%%%%%%%%%%%%%%%%%%%%%%%%%%%

\subsection{The Searching Phase} \label{sec:searching}

The searching phase of the algorithm is based on an alignment procedure\footnote{The alignment procedure described in this paper has been presented in a preliminary form in \cite{FP19}.} specifically designed to check, given two input strings $x$ and $y$ of the same length, if a set of translocations exists allowing to transform $x$ in $y$. The searching phase of our Algorithm \ref{fig:code3} simply consists in applying the alignment procedure between the pattern and any possible substring of the text of length $m$. 

\smallskip

\begin{algorithm}[h!] 
\caption{\label{fig:code3}Alignment Based Solution}
\begin{small}
\LinesNumbered
\SetNlSty{normal}{}{.}
\SetAlgoNoEnd
\KwIn{A pattern $x$ and a text $y$}
\KwOut{The number of approximate occurrences of $x$ in $y$}
\For{$j\leftarrow 0$ \emph{\textbf{to}} $n-m-1$} { \nllabel{alg3:loop1}
    $\Gamma^{(0)} \gets \{(0,0,$null,null$,0)\}$ \;
    \For{$i\leftarrow 1$ \emph{\textbf{to}} $m$} { 
        \ForEach{$(s_1,k_1,s_2,k_2,t) \in \Gamma^{(i-1)}$}  {  
             \If{$k_2=$ \emph{null}} { 
                \If{$x[i]=y[s+i]$} {
                    $\Gamma^{(i)} \gets \Gamma^{(i)} \cup \{ (i, \nil, \nil,\nil,t) \}$\; 
                }
                $r\gets \mu(y[s+i],i)$\;
                \While{$r\leq m$} { 
                    $\Gamma^{(i)} \gets \Gamma^{(i)} \cup \{ (i-1,0, r-1,t+1) \}$\; 
                    $r\gets \mu(y[s+i],r)$\;
                }
            }
            \Else {
                \If{$k_1=$ \emph{\textsf{null}}} { %$\Delta$ Case 2\\
                    \If{$x[s_2+k_2+1] = y[s+i]$} {
                        $\Gamma^{(i)} \gets \Gamma^{(i)} \cup \{ (s_1,k_1, s_2, k_2+1,t) \}$\; 
                    }
                    \Else {
                        $k_1\gets 0$\;
                    }
                }
                \If{$k_1\geq 0$} {
                    \While{$s_1+k_1<s_2$ \emph{\textbf{and}} $k_2>0$ \emph{\textbf{and}} $x[s_1+k_1+1] \neq y[s+i]$)} { %$\Delta$ Case 3.b
                        $b \gets 0$\;
                        \Repeat{$s_1+k_1+b<s_2$ \emph{\textbf{and}} $k_2-b>0$ \emph{\textbf{and}} $(k_1-s_1+1) \notin \phi(s_1+1,k_1+b)$} {
                            $b \gets b+\rho(s_1+1, s_2+k_2-b)$\;
                            $k_2 \gets k_2-b$\;
                            $k_1 \gets k_1+b$\;
                        }
                    }
                    \If{$s_1+k_1\geq s_2$ or $k_2\leq 0$} {
                        \textbf{break}\;
                    }
    
                    \If{$x[s_1+k_1+1] = y[s+i]$} { %$\Delta$ Case 3.a}\\
                        \If{$s_1+k_1 = s_2$ \emph{\textbf{and}} $(i,$ \emph{\nil, \nil, \nil}, t$) \notin \Gamma^{(i)}$} {
                            $\Gamma^{(i)} \gets \Gamma^{(i)} \cup \{ (i, \nil, \nil, \nil, t)\}$ \; 
                        }
                        \Else {
                            $\Gamma^{(i)} \gets \Gamma^{(i)} \cup \{ (s_1, k_1+1, s_2, k_2,t) \}$\;
                        }
                    }
                }
            }
        }
    }
    \If{$\Gamma^{(m)} \neq \emptyset$} {
        \textbf{Output} $j$\;
    }
}
\end{small}
\end{algorithm}

The alignment procedure finds a possible $utd$-match between two equal length strings. The pseudocode of Algorithm \ref{fig:code3} is tuned to process two strings $x$ and $y$, of length $m$, where translocations can take place only in $x$.
%The pseudocode of the approximate matching algorithm is presented in Figure \ref{fig:code2} and is tuned to process two strings $x$ and $y$, of length $m$ and $n$, respectively, in order to find any $utd$-match of $x$ in $y$.

In order to probe the details of the alignment procedure, let $x$ and $y$ be two strings of length $m$ over the same alphabet $\Sigma$. The procedure sequentially reads all characters of the string $y$, proceeding from left to right. While scanning it tries to evaluate all possible unbalanced translocations in $x$ which may be involved in the alignment between the two strings. 

%Assume that the first leftmost $i-1$ characters of $y$ have been scanned. 
We define a \emph{translocation attempt} at position $i$ of $y$, for $1\leq i \leq m$, as a quintuple of indexes, $(s_1, k_1, s_2, k_2, t)$, with all elements in $\{0,1,2, \ldots, m\} \cup \{\textsf{null}\}$ and where, referring to the string $x$, $s_1$ and $k_1$ pinpoints the leftmost position and the length of the first factor (the factor moved on the left), while $s_2$ and $k_2$ pinpoints the leftmost position and the length of the second factor (the factor moved on the right). 
In this context we refer to $s_1$ and $s_2$ as the \emph{key positions} of the translocation attempt.
The last value $t$ of the quintuple identify the cost of the attempt in terms of number of translocations.

In the special case where no translocation takes place in the attempt we assume by convention that $s_1 = i$ and $s_2=k_1=k_2 = \textsf{null}$.\footnote{We use the value \textsf{null} to indicate the length of an undefined string in order to discriminate it from the length of an empty string whose value is $0$ by definition.} 
During the execution of the algorithm for each translocation attempt, $(s_1, k_1, s_2, k_2,t)$,  at position $i$, the invariant given by the following lemma\footnote{In this context we assume that $s+\textsf{null} = s$, for any $s$.} holds.

\begin{lemma}\label{th0}
Let $y$ and $x$ be two strings of length $m$ over the same alphabet $\Sigma$. Let $\Gamma^{(i)}$ be the set of all translocation attempts computed by Algorithm \ref{fig:code3} during the $i$-th iteration. If $(s_1, k_1, s_2, k_2,t) \in \Gamma^{(i)}$ then it holds that:
\begin{enumerate}
\item[$(a)$] $i = s_1+k_1+k_2$;
\item[$(b)$] $x_i \trdmatch y_i$;
\item[$(c)$] if $s_2\neq$ \emph{\nil} then $x[s_1+1..s_1+k_1] = y[s_2+1..s_2+k_1]$;
\item[$(d)$] if $s_2\neq$ \emph{\nil} then $y[s_2+1..s_2+k_2] = y[s_1+1..s_1+k_2]$;
\end{enumerate}
\end{lemma}
\begin{proof}
We prove the statement by induction on the value of $i$. For the base case observe that when $i=0$ we have $\Gamma^{(0)} = \{0,\nil, \nil, \nil,0\}$ and the statements trivially hold. 

Suppose now that the statements hold for $i-1\geq 0$ and prove that they hold also for $i$. Let $(s_1, k_1, s_2, k_2,t)\in \Gamma^{(i-1)}$ be a translocation attempt computed at iteration $i-1$. By induction we have $s_1+k_1+k_2 = i-1$.  We can distinguish the following three different cases.

If we are in Case 1 and $x[i]=y[i]$ then $(i, \nil, \nil, \nil, t)$ is added to $\Gamma^{(i)}$ which trivially satisfies all the statements (\emph{a-d}). Moreover, for each $j>i$ such that $x[j]=y[i]$ the new translocation attempt $(i-1,\nil,j-1,1,t+1)$ is added to $\Gamma^{(i)}$. In this last case we have that the statement ($b$) still holds and we have also that $s_1 + k_1+k_2 = s_1+1 = i-1+1 = i$, so the statement ($a$) holds too. 

If we are in Case 2 the right factor can be extended by a single character to the right adding $(s_1, k_1, s_2, k_2+1,t)$ to $\Gamma^{(i)}$. Also in this case we have $s_1+k_1+k_2+1 = i-1+1=i$ and the property is satisfied.

Finally, if we are in Case 3 the left factor can be extended by a single character to the right adding $(s_1, k_1+1, s_2, k_2,t)$ to $\Gamma^{(i)}$ getting again $s_1+k_1+1+k_2 = i-1+1 = i$.
Observe also that if we close a translocation attempt (when $k_1+1=s_2$) the attempt $(i, \nil, \nil, \nil,t)$ is added to $\Gamma^{(i)}$, which trivially satisfies the property.
\end{proof}

For each $1\leq i \leq m$, we define $\Gamma^{(i)}$ as the set of all translocation attempts tried for the prefix $y[1 .. i]$, and set $\Gamma^{(0)} = \{(0,\textsf{null},\textsf{null},\textsf{null}),0\}$.  

However we can prove that it is not necessary to process all possible translocation attempts. Some of them, indeed, leads to detect the same $utd$-matches so that they can be skipped.
Specifically it can be proved that the following Lemma holds.
%\footnote{The proof of Lemma \ref{ob:1} is given in Lemma \ref{th1}}.

\begin{lemma}\label{th1}
Let $y$ and $x$ be two strings of length $m$ over the same alphabet $\Sigma$. Let $s \sqsubseteq y$ and $u \sqsubseteq x$ such that $|s|=|x|$ and $s \trdmatch u$. Moreover assume that
\begin{itemize}
\item [\emph{(i)}] $s.w.z \sqsubseteq y$ and $u.z.w \sqsubseteq x$
\item [\emph{(ii)}] $s.w'.z \sqsubseteq y$ and $u.z.w' \sqsubseteq x$
\end{itemize}
with $|z|>0$ and $|w'|>|w|>0$. If we set $i = |s.w.z|$ and $j=|s.w'.z|$ then we have $x[i+1 .. j] \trdmatch y[i+1 .. j]$.
\end{lemma}
\begin{proof}
We will consider the following three cases, which depend on the length of $w'$ in relation with the length of $w.z$.
The three cases discussed below are depicted in Figure \ref{fg12}, cases (a), (b) and (c), respectively.
\begin{itemize}
\item{Case (a)}: suppose first  $|w'| \geq |w|+|z|$.
Refering to $y$, we have that $w.z$ is a prefix of $w'$ and we can write  $w' = w.z.a$ for some $a\in\Sigma^*$, with $|a|\geq 0$. In the specific case where $|w'|=|w|+|z|$ we have $w'=w.z$ and $a=\varepsilon$. Thus by (ii) we can write $s.w.z.a.z \sqsubseteq y$ where  $y[i+1 .. j]=a.z$. Similarly we have $u.z.w' = u.z.w.z.a \sqsubseteq x$, so that  $x[i+1 .. j] = z.a$. Thus it follows that $x[i+1 .. j] \trdmatch y[i+1 .. j]$.

\item{Case (b)}: suppose now  $|w|+|z|/2 \leq |w'| < |w|+|z|$.
Referring to $y$, observe that $z$ has a border $b$, with $|b|\leq |z|/2$. Thus we can write  $z = b.a.b$ for some $a\in\Sigma^*$, with $|a|\geq 0$. In the specific case where $|b|=|z|/2$ we have $z=b.b$ and $a=\varepsilon$.
Observe that $w$ is a prefix of $w'$. Let $c$ be the suffix of $w'$ such that $w' = w.c$. Since $c$ is also the prefix of $z$ of length $|w'|-|w| = |z|-|b|$, it follows $c = b.a$. Thus we have $y[i+1 .. j] = a.b$ and $x[i+1 .. j]= b.a$, and it trivially follows that $x[i+1 .. j] \trdmatch y[i+1 .. j]$.

\item{Case (c)}: suppose finally  $|w'| < |w|+|z|/2$.
 Also in this case $z$ has a border $b$. However now we have $|b|> |z|/2$ and $z$ is of the form $z=a.b=b.a$, since $|a|<|b|$ and $a$ is a border of both $z$ and $b$.
 Referring to $y$, we can observe also that $a \sqsupseteq w'$ and that $|w'|-|w| = |a|$. Thus we have $x[i+1 .. j] = y[i+1 .. j] = a$. proving the Lemma.
\end{itemize}
\end{proof}

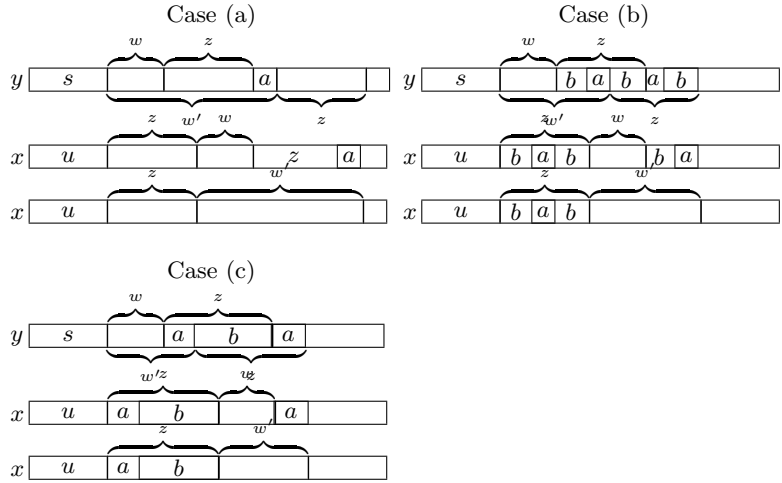
\begin{figure}[t]
\begin{center}
\setlength{\unitlength}{0.006\textwidth}
\setlength{\fboxrule}{1mm}

\begin{tabular}{lr}
\begin{picture}(70,45)
 % caso 1
 \put(35,40){\makebox(3,3)[c]{\small Case (a)}}
 \put(0,28){\makebox(30,3)[l]{$y$}}
 \put(2,28){  $                    \framebox(14,4)[c]{$s$}
		                         \hspace{-1pt}\overbrace{\framebox(10,4)[c]{}}^{w}
		                         \hspace{-1pt}\overbrace{\framebox(16,4)[c]{}}^{z}
		                        \hspace{-1pt}\framebox(4,4)[c]{$a$}
		                        \hspace{-1pt}\underbrace{\framebox(16,4)[c]{}}_{z}
		                        \hspace{-1pt}\framebox(4,4)[c]{}$}
 \put(16,28){  $                    \underbrace{\makebox(31,4)[c]{}}_{w'}$}
 %\put(48,28){ \makebox(16,4)[c]{}\makebox(4,4)[c]{$a$}\framebox(16,4)[c]{}}

 \put(0,14){\makebox(30,3)[l]{$x$}}
 \put(2,14){ $	\framebox(14,4)[c]{$u$}
                         \hspace{-1pt}\overbrace{\framebox(16,4)[c]{}}^{z}
                         \hspace{-1pt}\overbrace{\framebox(10,4)[c]{}}^{w}
                       \hspace{-1pt}\framebox(24,4)[c]{}$}
 %\put(38,14){ \makebox(6,4)[c]{$b$}\framebox(4,4)[c]{$a$}\makebox(6,4)[c]{$b$}}
 \put(42,14){ \makebox(16,4)[c]{$z$}\framebox(4,4)[c]{$a$}}

 \put(0,4){\makebox(30,3)[l]{$x$}}
 \put(2,4){ $	\framebox(14,4)[c]{$u$}
                         \hspace{-1pt}\overbrace{\framebox(16,4)[c]{}}^{z}
                         \hspace{-1pt}\overbrace{\framebox(30,4)[c]{}}^{w'}
                       \hspace{-1pt}\framebox(4,4)[c]{}$}
 %\put(38,4){ \makebox(6,4)[c]{$b$}\framebox(4,4)[c]{$a$}\makebox(6,4)[c]{$b$}}
\end{picture} &
%\end{subfigure}
%%%%%%%%%%%
%\begin{subfigure}{.48\textwidth}
\begin{picture}(70,45)
 % caso 2
 \put(35,40){\makebox(3,3)[c]{\small Case (b)}}
 \put(0,28){\makebox(30,3)[l]{$y$}}
 \put(2,28){  $                 \framebox(14,4)[c]{$s$}
		                         \hspace{-1pt}\overbrace{\framebox(10,4)[c]{}}^{w}
		                         \hspace{-1pt}\overbrace{\framebox(16,4)[c]{}}^{z}
		                        \hspace{-1pt}\framebox(24,4)[c]{}$}
 \put(16,28){  $                    \underbrace{\makebox(20,4)[c]{}}_{w'}
 					\hspace{-1pt}\underbrace{\makebox(16,4)[c]{}}_{z} $}
 \put(26,28){ \makebox(6,4)[c]{$b$}\framebox(4,4)[c]{$a$}\makebox(6,4)[c]{$b$}\makebox(4,4)[c]{$a$}\framebox(6,4)[c]{$b$}}

 \put(0,14){\makebox(30,3)[l]{$x$}}
 \put(2,14){ $	\framebox(14,4)[c]{$u$}
                         \hspace{-1pt}\overbrace{\framebox(16,4)[c]{}}^{z}
                         \hspace{-1pt}\overbrace{\framebox(10,4)[c]{}}^{w}
                       \hspace{-1pt}\framebox(24,4)[c]{}$}
 \put(16,14){ \makebox(6,4)[c]{$b$}\framebox(4,4)[c]{$a$}\makebox(6,4)[c]{$b$}}
 \put(42,14){ \makebox(6,4)[c]{$b$}\framebox(4,4)[c]{$a$}}

 \put(0,4){\makebox(30,3)[l]{$x$}}
 \put(2,4){ $	\framebox(14,4)[c]{$u$}
                         \hspace{-1pt}\overbrace{\framebox(16,4)[c]{}}^{z}
                         \hspace{-1pt}\overbrace{\framebox(20,4)[c]{}}^{w'}
                       \hspace{-1pt}\framebox(14,4)[c]{}$}
 \put(16,4){ \makebox(6,4)[c]{$b$}\framebox(4,4)[c]{$a$}\makebox(6,4)[c]{$b$}}
\end{picture}\\
%\end{subfigure}
%%%%%%%%%%%
%\begin{subfigure}{.48\textwidth}
\begin{picture}(70,45)
 % caso 3
 \put(35,40){\makebox(3,3)[c]{\small Case (c)}}
 \put(0,28){\makebox(30,3)[l]{$y$}}
 \put(2,28){  $                 \framebox(14,4)[c]{$s$}
		                         \hspace{-1pt}\overbrace{\framebox(10,4)[c]{}}^{w}
		                         \hspace{-1pt}\overbrace{\framebox(19.5,4)[c]{}}^{z}
		                        \hspace{-1pt}\framebox(20,4)[c]{}$}
 \put(16,28){  $                    \underbrace{\makebox(16,4)[c]{}}_{w'}
 					\hspace{-1pt}\underbrace{\makebox(20,4)[c]{}}_{z} $}
 \put(26,28){ \makebox(6,4)[c]{$a$}\framebox(14,4)[c]{$b$}\framebox(6,4)[c]{$a$}}

 \put(0,14){\makebox(30,3)[l]{$x$}}
 \put(2,14){ $	\framebox(14,4)[c]{$u$}
                         \hspace{-1pt}\overbrace{\framebox(20,4)[c]{}}^{z}
                         \hspace{-1pt}\overbrace{\framebox(10,4)[c]{}}^{w}
                       \hspace{-1pt}\framebox(20,4)[c]{}$}
 \put(16,14){ \makebox(6,4)[c]{$a$}\framebox(14.2,4)[c]{$b$}}
 \put(46.5,14){ \framebox(6,4)[c]{$a$}}

 \put(0,4){\makebox(30,3)[l]{$x$}}
 \put(2,4){ $	\framebox(14,4)[c]{$u$}
                         \hspace{-1pt}\overbrace{\framebox(20,4)[c]{}}^{z}
                         \hspace{-1pt}\overbrace{\framebox(16,4)[c]{}}^{w'}
                       \hspace{-1pt}\framebox(14,4)[c]{}$}
 \put(16,4){ \makebox(6,4)[c]{$a$}\framebox(14.2,4)[c]{$b$}}
\end{picture} & \\
%\end{subfigure}
%
\end{tabular}
\end{center}
\caption{\label{fg11}The three cases of Lemma \ref{th1}. In {Case (a)} we suppose  $|w'| \geq |w|+|z|$; in {Case (b)} it is supposed that  $|w|+|z|/2 \leq |w'| < |w|+|z|$; finally in {Case (c)} we suppose   $|w'| < |w|+|z|/2$.}
\end{figure}

The procedure iterates on the values of $i$, for $1\leq i \leq m$, while scanning the characters of $y$, and during the $i$-th iteration it computes the set $\Gamma^{(i)}$ from $\Gamma^{(i-1)}$.
For each translocation attempt $(s_1, k_1, s_2, k_2,t) \in \Gamma^{(i-1)}$ we distinguish the following three cases (depicted in Fig. \ref{fg11}):
\begin{itemize}

\item{Case 1} ($s_2=k_1=k_2=$ \textsf{null})\\
This is the case where no unbalanced translocation is taking place (line 6). Thus we simply know that $x_{i-1} \trdmatch y_{i-1}$. If $x[i]=y[i]$ the match is extended of one character by adding the attempt $(s_1+1$,\textsf{null},\textsf{null},\textsf{null}$,0)$ to $\Gamma^{i}$ (line 7). Alternatively, when possible, new translocation attempts are started (lines 9-11). Specifically for each occurrence of the character $y[i]$ in $x$, at a position $r$ next to $s_1$, a new right factor $u_r$ is attempted starting at position $r$ (line 10) by extending $\Gamma^{(i)}$ with the attempt $(s_1, 0, r-1,1,t+1)$.

\begin{figure}[t]
\begin{center}
\setlength{\unitlength}{0.006\textwidth}
\setlength{\fboxrule}{1mm}

\begin{tabular}{lr}
%\begin{subfigure}{.48\textwidth}
\begin{picture}(60,40)
 % caso 1
 \put(35,27){\makebox(3,3)[c]{\small Case (1)}}
 \put(0,16){\makebox(30,3)[l]{$y$}}
 \put(2,16){  $\overbrace{\framebox(8,4)[c]{}}^{i-1}$\framebox(54,4)[c]{}}
%\put(12,21){\makebox(3,3)[l]{$i$}}
\put(12,16.5){\makebox(3,3)[l]{$\bullet$}}
 \put(0,4){\makebox(30,3)[l]{$x$}}
 \put(2,4){  $\overbrace{\framebox(8,4)[c]{}}^{i-1}$\framebox(54,4)[c]{}}
%\put(12,9){\makebox(3,3)[l]{$i$}}
\put(12,4.5){\makebox(3,3)[l]{$\bullet$}}
\end{picture} &
%\end{subfigure}
%
%\begin{subfigure}{.48\textwidth}
\begin{picture}(70,40)
 % caso 2
 \put(35,27){\makebox(3,3)[c]{\small Case (2)}}
 \put(0,16){\makebox(30,3)[l]{$y$}}
 \put(2,16){  $\overbrace{\framebox(8,4)[c]{}\framebox(26,4)[c]{$u_r$}}^{i-1}$\framebox(28,4)[c]{}}
\put(38,16.5){\makebox(3,3)[l]{$\bullet$}}
 \put(0,4){\makebox(30,3)[l]{$x$}}
 \put(2,4){ 	\framebox(8,4)[c]{}\framebox(22,4)[c]{}$\underbrace{\framebox(26,4)[c]{$u_r$}}_{k_2}$\framebox(6,4)[c]{}}
\put(60,4.5){\makebox(3,3)[l]{$\bullet$}}
\put(8,8){\makebox(3,3)[l]{$s_1$}}
\put(30,8){\makebox(3,3)[l]{$s_2$}}
\end{picture}\\
%\end{subfigure}
%
%\begin{subfigure}{.48\textwidth}
\begin{picture}(70,55)
 % caso 3
 \put(35,42){\makebox(3,3)[c]{\small Case (3.a)}}
 \put(0,32){\makebox(30,3)[l]{$y$}}
 \put(2,32){  $\overbrace{\framebox(8,4)[c]{}\framebox(26,4)[c]{$u_r$}\framebox(10,4)[c]{$u_{\ell}$}}^{i-1}$\framebox(18,4)[c]{}}
\put(48,32.5){\makebox(3,3)[l]{$\bullet$}}
 \put(0,16){\makebox(30,3)[l]{$x$}}
 \put(2,16){ 	\framebox(8,4)[c]{}$\underbrace{\framebox(10,4)[c]{$u_{\ell}$}}_{k_1}$\framebox(12,4)[c]{}$\underbrace{\framebox(26,4)[c]{$u_r$}}_{k_2}$\framebox(6,4)[c]{}}
\put(22,16.5){\makebox(3,3)[l]{$\bullet$}}
\put(8,20){\makebox(3,3)[l]{$s_1$}}
\put(30,20){\makebox(3,3)[l]{$s_2$}}
\end{picture}&
%\end{subfigure}
%
%\begin{subfigure}{.48\textwidth}
\begin{picture}(70,55)
 % caso 4
 \put(35,42){\makebox(3,3)[c]{\small Case (3.b)}}
 \put(0,32){\makebox(30,3)[l]{$y$}}
 \put(3,32){$\overbrace{\framebox(8,4)[c]{}\framebox(26,4)[c]{$u_r$}\framebox(10,4)[c]{$u_{\ell}$}}^{i-1}$\framebox(18,4)[c]{}}
 \put(3,32){\framebox(8,4)[c]{}$\underbrace{\makebox(22,4)[c]{}}_{u'_{\ell}}\hspace{-1pt}\underbrace{\makebox(14,4)[c]{}}_{u'_r}$}
\put(48,32.5){\makebox(3,3)[l]{$\bullet$}}
 \put(0,16){\makebox(30,3)[l]{$x$}}
 \put(2,16){ 	\framebox(8,4)[c]{}$\underbrace{\framebox(10,4)[c]{$u_{\ell}$}}_{k_1}$\framebox(12,4)[c]{}$\underbrace{\framebox(26,4)[c]{$u_r$}}_{k_2}$\framebox(6,4)[c]{}}
%\put(22,16.5){\makebox(3,3)[l]{$\bullet$}}
\put(8,20){\makebox(3,3)[l]{$s_1$}}
\put(30,20){\makebox(3,3)[l]{$s_2$}}
 \put(0,4){\makebox(30,3)[l]{$x$}}
 \put(2,4){ 	\framebox(8,4)[c]{}$\underbrace{\framebox(8,4)[c]{$w$}\framebox(10,4)[c]{$u_{\ell}$}}_{k_1}$\framebox(4,4)[c]{}$\underbrace{\framebox(18,4)[c]{$u'_r$}}_{k_2}$\framebox(8,4)[c]{$w$}\framebox(6,4)[c]{}}
\put(30,4.5){\makebox(3,3)[l]{$\bullet$}}
\put(8,8){\makebox(3,3)[l]{$s_1$}}
\put(30,8){\makebox(3,3)[l]{$s_2$}}
\end{picture}\\
%\end{subfigure}
%
\end{tabular}
\end{center}
\caption{\label{fg12}Three cases of Algorithm \ref{fig:code3} while processing the translocation attempt $(s_1,k_1,s_2,k_2)\in\Gamma^{(i-1)}$ in order to extend it by charcater $y[i]$. Character $y[i]$ and its counterpart in $x$ are depicted by a bullet symbol. Case (1): $x_{i-1} \trdmatch y_{i-1}$ and $x[i]=y[i]$, then the match is extended of one character; Case (2):  the right factor $u_r$ is currently going to be recognized and $x[s_2+k_2+1] = y[i]$, then the right factor $u_r$ can be extended of one character;  Case (3.a): the left factor $u_{\ell}$ can be extended of one character to the right; Case (3.b): the right factor $u_{\ell}$ cannot be extended, then we try to transfer a suffix $w$ of $u_r$ to the prefix position of $u_{\ell}$, reducing the length $k_2$ and extending the length $k_1$ accordingly ($w$ is a suffix of $u_r$ and also a prefix of $x[s_1+1 .. s_2]$ and, in addition, we can move $u_{\ell}$ to the right of $|w|$ position along the left factor).}
\end{figure}
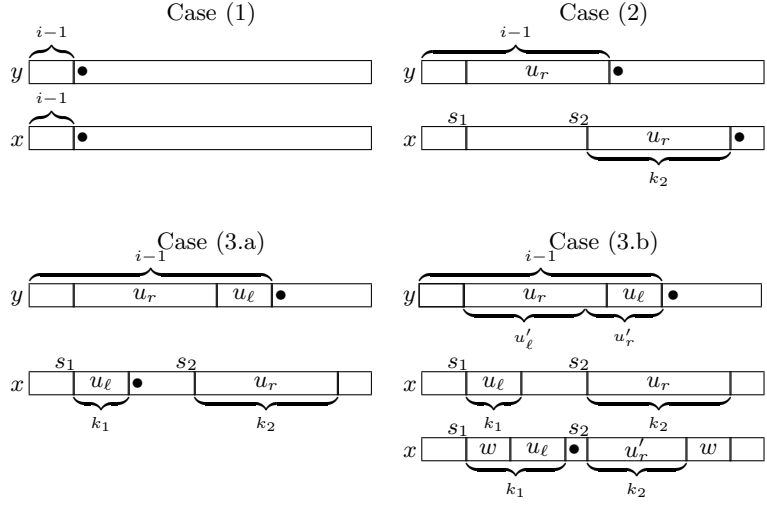

\item{Case 2} ($k_1=0$ and $k_2 > 0$)\\
This is the case where an unbalanced translocation is taking place and the right factor $u_r$ is currently going to be recognized (line 14). Specifically we know that $x[s_2+1 .. s_2+k_2] = y[i-k_2 .. i-1]$ and that $x[1 .. s_1] \trdmatch y[1 .. i-k_2-1]$.
If $x[s_2+k_2+1] = y[i]$ the right factor $u_r$ can be extended of one character to the right, thus $\Gamma^{i}$ is extended by adding the attempt $(s_1,k_1,s_2,k_2+1,t)$ (line 15).
Otherwise, if $x[s_2+k_2+1] \neq y[i]$, the right factor $u_r$ cannot be extended further, thus we start recognizing the left factor $u_{\ell}$. Specifically, in this last case, we update $k_1$ to $0$ (line 17) and move to the following Case 3.

\item{Case 3} ($k_1\geq 0$)\\
This is the case where an unbalanced translocation is taking place,  the right factor $u_r$ has been already recognized and we are attempting to recognize the left factor $u_{\ell}$. Specifically we know that $x[s_1+1 .. s_1+k_1] = y[i-k_1 .. i-1]$, $x[s_2+1 .. s_2+k_2] = y[i-k_1-k_2 .. i-k_1-1]$  and that $x[1 .. s_1] \trdmatch y[1 .. i-k_1-k_2-1]$. We distinguish two sub-cases:
	\begin{itemize}
	\item {Case 3.a} ($x[s_1+k_1+1] = y[i]$)\\
	In this case the left factor $u_{\ell}$ can be extended of one character to the right (line 24). Thus if the left factor has been completely recognized, i.e. if $s_1+k_1 = s_2$, $\Gamma^{i}$ is extended by adding the attempt $(s_1+k_1+k_2, \nil, \nil, \nil),t$ (lines 29-30) which indicates that $x_i \trdmatch y_i$. Otherwise $\Gamma^{(i)}$ is extended by adding the attempt $(s_1, k_1+1, s_2, k_2,t)$ (line 32).
	\item {Case 3.b} ($x[s_1+k_1+1] \neq y[i]$)\\
	In this case the right factor $u_{\ell}$ cannot be extended. Before quitting the translocation attempt we try to find some new factors rearrangements on the same key positions, $s_1$ and $s_2$, but with different lengths, $k_1$ and $k_2$. Specifically we try to transfer a suffix $w$ of $u_r$ to the prefix position of $u_{\ell}$, reducing the length $k_2$ and extending the length $k_1$ accordingly. This can be done only if we find a suffix $w$ of $u_r$ which is also a prefix of $x[s_1+1 .. s_2]$ and, in addition, we can move $u_{\ell}$ to the right of $|w|$ position along the left factor. More formally, if we assume that $|w|=b$ we must have:
		\begin{enumerate}
			\item $|w|< |u_r|$, or rather $b < k_2$  (indicating that $w$ is a proper suffix of $u_r$);
			\item $b \in \phi(s_1+1, s_2+k_2)$  (indicating that $w$ is a border of $x[s_1+1 .. s_2+k_2]$);
			\item $(k_1-s_1+1) \in \phi(s_1+1,k_1+b)$  (indicating that $u_{\ell}$ is a border of $x[s_1+1 .. s_1+k_1+|w|]$);
			\item $s_1+k_1+|w| < s_2$  (indicating that the updated $u_{\ell}$ does not overflow onto $u_r$);
			\item $x[s_1+k_1+|w|+1] = y[i]$  (indicating that the updated $u_{\ell}$ can be extended by $y[i]$)
		\end{enumerate}
	\end{itemize}

\end{itemize}

%%%%%%%%%%%%%%%%%%%%

\subsection{Worst-case time and space analysis}\label{sec:complexity}
In this section we discuss the worst-case time and space analysis of Algorithm \ref{fig:code3} presented in the previous section. %In particular, we will refer to the
%implementation reported in Figure \ref{fig:code1}.

Let $x$ be a pattern of length $m$ and $y$ be a text of length $n$ over the same alphabet $\Sigma$ and assume to run Algorithm \ref{fig:code3} for searching all $\delta$-bounded approximate occurrences of $x$ in $y$.

Regarding the space analysis, as stated in Section \ref{sec:preprocessing} we need $O(m\sigma)$ additional space to maintain the next position function, $O(m^3)$-space to maintain the border set function and $O(m^2)$-space to maintain the shortest border function. During the searching phase we need also $O(m)$ space to maintain the key position table $\gamma$. 
Thus the overall space complexity of the algorithm is $O(m^3)$.

Regarding the time analysis, we first consider the time complexity required by a single iteration of the main loop of line \ref{alg3:loop1} of the algorithm. Specifically let $\Gamma^{(i)}$ be the the set of all translocation attempts computed at iteration $i$, for $0\leq i \leq m$. 

First of all we observe that each $\Gamma^{(i)}$ contains at most one translocation attempt with $k_2=\textsf{null}$ (i.e. of the form $(s_1, \nil, \nil, \nil),t$).  We put $\Gamma^{(0)} = \{(0,\textsf{null},\textsf{null},\textsf{null},0)\}$ (line 2), thus the statement holds for $i=0$. Observe that, if $i>0$ a translocation attempt of the form $(s_1, \nil, \nil, \nil,t)$ can be added to $\Gamma^{(i)}$ only in line 7 or in line 30. However by condition at line 29, if it is added to $\Gamma^{(i)}$ in line 7, it cannot be added again in line 30.

\smallskip

We now prove that the total number of translocation attempts processed during a sngle execution of the main loop of the algorithm is bounded by $m^3$. More formally we have
\begin{equation}\label{eq:m2}
	\sum_{i=0}^{m}|\Gamma^{(i)}|\leq m^3.
\end{equation}
To prove that equation (\ref{eq:m2}) holds observe that new translocation attempts are added to $\Gamma^{(i)}$ only when we are in Case 1. When we are in Case 2 or in Case 3 a translocation attempt is simply rearranged by extending the right factor (Case 2) or the left factor (Case 3). As observed above only one translocation attempt in $\Gamma^{(i)}$ is in Case 1 and the \textsf{while} cycle of line 9 can add at most $m-i$ new translocation attempts to $\Gamma^{(i+1)}$. In the worst case each translocation attempt added to $\Gamma^{(i+1)}$ will be closed only at iteration $m$, thus it will be extended along the sets $\Gamma^{(j)}$, for $j>i$. Thus the overall contribute of each translocation attempt added to $\Gamma^{i+1}$ is $m-i$.

Thus the total number of translocation attempts processed during a single iteration of the main loop of the algorithm is bounded by
$$
	\begin{array}{rl}
	\displaystyle\sum_{i=0}^{m}|\Gamma^{(i)}| & \leq  \displaystyle1+\sum_{i=1}^{m} (m-i)^2 =\\[0.5cm]
		&\displaystyle = 1+\sum_{i=1}^{m} m^2 - \sum_{i=1}^{m} 2im +\sum_{i=1}^{m} i^2=\\[0.5cm]
	 	&\displaystyle = 1+m^3 + \frac{m(m+1)(2m+1)}{6} - \frac{m(m+1)}{2} = \\[0.5cm]
		&\displaystyle = \frac{1}{3}m^3 - \frac{1}{2}m^2 + \frac{1}{3} \leq m^3 
	\end{array}
$$

\smallskip

Finally we observe that each translocation attempt in Case 2 and Case 3.a is processed in constant time, during the execution of Algorithm \ref{fig:code3}.  
A translocation attempt in Case 1 my be processed in $O(m-i)$ worst case time. However the overall contribution of the while cycle at line 9 is at most $O(m^2)$ since, as observer above, there is a single translocation attempt in Case 1 for each $\Gamma^{(i)}$.

For a translocation attempt in Case 3, observe that at each execution of line 22 the value of $b$ is increased of at most $1$. Then in line 23 we decrease $k_2$ by $b$. Since the value of $k_2$ is increased only in line 15, this implies that overall number of times the \textsf{while} cycle of line 21 is executed is bounded by $k_2$, which is at most $m$. Thus the overall contribution given by the while cycle of line 19 is $O(m^3)$.

We can conclude that the overall time complexity of a single iteration of Algorithm \ref{fig:code3} is $\bigO(m^3)$ and that the overall complexity of the algorithm is therefore $\bigO(nm^3)$.

\subsection{Solving other variants of the problem}
Algorithm 3, as proposed in the previous section, solves variant (c) of the approximate string matching problem allowing for non overlapping unbalanced translocations. Thus it finds all positions $s$ in $y$ such that $x$ has a $\delta$-bounded approximate occurrence in $y$ at position $s$.
In this section we briefly discuss how to slightly modify the algorithm to obtain a solution for variants (a), (b) and (d) respectively.

For solving \textbf{Variant (a)} of the problem a translocation attempt at a given position of $y$ can be simply represented as a quadruple of indexes $(s_1,k_1,s_2,k_2)$. Thus we don't need to maintain the cost $t$ of the attempt in terms of number of translocations. The resulting algorithm could be simplified accordingly, although this simplification doesn't lead to a reduction in terms of time or space complexity.

\textbf{Variant (b)} of the problem asks to find the number of all $\delta$-bounded approximate occurrences of $x$ in $y$.  In this case it is enough to modify Algorithm 3 in order to count the matching positions while they are given in line 34.  Thus the solution maintains the same time and space complexity.

Finally, \textbf{Variant (d)} of the problem asks to find, for each position $s$ in $y$, the number of distinct $\delta$-bounded approximate occurrences of $x$ in $y$ at position $s$. As i the previous case it is enough to maintain a counter which must be increased, in line 34, by the size of the set $\Gamma^{m}$. Also in this case the solution maintains the same time and space complexity as Algorithms 3.

%%%%%%%%%%%%%%%%%%%%%%%%%%%

\section{Conclusions and Future Works}\label{sec:conclusions}
We presented a three solutions for the approximate string matching problem allowing for unbalanced translocations of adjacent factors working in $\bigO(nm^3)$ worst case time using $\bigO(m^2)$-space.
To the best of our knowledge this is the first paper addressing this kind of problem.

Although our three solutions have the same time and space complexity Algorithm 2 and Algorithm 3 have some special features.
Specifically we prove that Algorithm 2 has a $\bigO(n\log_{\sigma}^2m)$ average case time complexity, for alphabets of size $\sigma\geq 4$.
In addition Algorithm 3 uses a constructive approach which could be more efficient in practice  and which can be easily optimized.  It turns out, indeed, by our preliminary experimental results that our solution has a sub-quadratic behaviour in practical cases. This suggests us to focus our future works on an accurate analysis of the algorithm's complexity in the average case.

%%%%%%%%%%%%%%%%%%%%%%%%%%%%%%%%%%%%%%%%%%%%%%%%%%%%%%

\end{document}